\documentclass[final]{elsarticle}

\usepackage{hyperref}
\usepackage{amssymb,amsmath,url,amsmath,graphicx,mathrsfs,color}
\usepackage{algorithm}
\usepackage{algpseudocode}
\usepackage{multirow}
\usepackage[table]{xcolor}
\algnewcommand\algorithmicinput{\textbf{Initialization:}}
\algnewcommand\Initialize{\item[\algorithmicinput]}

\journal{Theoretical Computer Science}

\def\CC{\mathbb C}
\def\RR{\mathbb R}

\def\cF{\mathcal F}

\def\cI{\mathcal I}

\def\cN{\mathcal N}

\def\cO{\mathcal O}

\def\cW{\mathcal W}

\newcommand{\xcp}{ x^{ \text{\tiny  \rm cx}}}

\newcommand{\xlp}{ x^{ \text{\tiny \rm  lp}}}

\def\bzero{\mathbf 0}

\newcommand{\raf}[1]{(\ref{#1})}

\newcommand{\cT}{\ensuremath{\mathcal{T}} }

\newcommand{\OPT}{\ensuremath{\textsc{Opt}}}

\newcommand{\re}{\ensuremath{\mathrm{Re}}}
\newcommand{\im}{\ensuremath{\mathrm{Im}}}

\newtheorem{theorem}{Theorem}
\newtheorem{lemma}{Lemma}

\newtheorem{definition}{Definition}
\newtheorem{proposition}{Proposition}
\newtheorem{remark}{Remark}

\newtheorem{innercustomlem}{Lemma}
\newenvironment{customlem}[1]
{\renewcommand\theinnercustomlem{#1}\innercustomlem}
{\endinnercustomlem}

\begin{document}

\begin{frontmatter}

\title{Complex-demand  Scheduling Problem  with Application in \mbox{Smart Grid}\footnote{This paper appears in Theoretical Computer Science (DOI: 10.1016/j.tcs.2018.08.023). A preliminary version appeared in the 22nd International Conference on Computing and Combinatorics, COCOON 2016, Ho Chi Minh City, Vietnam, August 2-4, 2016.}}


\author{Majid Khonji} 
\ead{majid.khonji@ku.ac.ae}
\author{Areg Karapetyan}  
\ead{areg.karapetyan@ku.ac.ae}
\author{Khaled Elbassioni}
\address{Masdar Institute, Khalifa University of Science and Technology, Abu Dhabi, UAE}
\ead{khaled.elbassioni@ku.ac.ae}

\author{Sid Chi-Kin Chau}
\address{Research School of Computer Science, Australian National University, Canberra, Australia}
\ead{sid.chau@anu.edu.au}


\begin{abstract}

We consider the problem of scheduling complex-valued demands over a discretized time horizon. Given a set of users, each user is associated with a set of demands representing different power consumption preferences. A demand is represented by a complex number, a time interval, and a utility value obtained if it is satisfied. At each time slot, the magnitude of the total selected demands should not exceed a given generation capacity. This naturally captures the supply constraints in alternating current (AC) electric systems. In this paper, we consider maximizing the aggregate user utility subject to power supply limits over a time horizon. We present approximation algorithms characterized by the maximum angle $\phi$ between any two complex-valued demands. More precisely, a PTAS is presented for the case $\phi \in [0,\tfrac{\pi}{2}]$, a bi-criteria FPTAS for $\phi \in [0,{\pi} \mbox{-} \varepsilon]$ for any polynomially small $\varepsilon$, assuming the number of time slots in the discretized time horizon is a constant. Furthermore, if the number of time slots is part of the input, we present a reduction to the real-valued unsplittable flow problem on a path with only a constant  approximation ratio. Finally, we present a practical greedy algorithm for the single time slot case with an approximation ratio of $\tfrac{1}{2}\cos \frac{\phi}{2}$ and a running time complexity of only ${O}(N\log N)$, $N$ standing for the aggregate number of user demands, which can be implemented efficiently in practice.


%
%
%
%
\end{abstract}

\begin{keyword}
Algorithms, Scheduling, Smart Grid, Unsplittable Flow, Knapsack
\end{keyword}

\end{frontmatter}


\section{Introduction}
A key aspect of the emerging smart grid is to modulate users' electricity consumption around the available power supply. A microgrid could run short of power due to emergency conditions, high electricity purchase price in the bulk market, or volatility of renewable sources. In such cases, consumers' deferrable loads, such as dish washers and electric vehicles, can be scheduled according to the grid's operational or economic conditions. This, in fact, models the day-ahead electric market at the distribution network whereby customers provide their deferrable demand preferences along with the amount they are welling to pay, and the grid operator decides the best allocation.

Although resource allocation and scheduling mechanisms have been well-studied in various systems from transportation to communication networks, the rise of the smart grid presents a new range of algorithmic problems, which are a departure from these systems. One focal difference is the presence of periodic time-varying entities (e.g., current, power, voltage) in AC electric systems, which are often expressed in terms of non-positive real, or even complex numbers. In power terminology \cite{GS94power}, the real component of the complex number is called the {\em active} power, the imaginary is known as {\em reactive } power, and the magnitude as {\em apparent} power. For example, purely resistive appliances have positive active power and zero reactive power. Appliances and instruments with capacitive or inductive components have non-zero reactive power, depending on the phase lag with the input power. Machinery, such as in factories, has large inductors, and hence has positive power demand. On the contrary, shunt-capacitor equipped electric vehicle charging stations can generate reactive power. 

We consider a variable power generation capacity over a discrete time horizon. Every user of the smart grid is associated with a set of demand preferences, wherein a demand is represented by a complex-valued number, a time interval at which it should be supplied, and a utility value obtained if it is satisfied. Some demands are inelastic (i.e., indivisible) in a sense that are either fully satisfied, or completely dropped. At each time slot, the magnitude of the total satisfied demands among all different preferences  should not exceed the current net generation capacity of the grid. This captures the variation in supply constraints over time in alternating current (AC) electric systems, and allows to model the demand response management in power systems\cite{DR09}.


Conventionally, demands in AC systems are represented by complex numbers in the first and fourth quadrants of the complex plane. We note that our problem is invariant, when the arguments of all demands are shifted by the same angle. 
For convenience, we assume the demands are rotated such that one of them is aligned along the positive real axis. 
In realistic setting of power systems, the active power demand is positive, but the power factor (i.e., the cosine of the demand's argument) is bounded from below by a certain threshold, which is equivalent to restricting the argument of complex-valued demands.

We present approximation algorithms characterized by the maximum angle $\phi$ between any two complex-valued demands. More precisely,  we present a PTAS for the case $\phi \in [0,\tfrac{\pi}{2}]$, a bi-criteria FPTAS for  $\phi \in [0,{\pi} \mbox{-} \varepsilon]$ for any polynomially small $\varepsilon$, assuming the number of time slots in the discretized time horizon is constant. 
Furthermore, if the number of time slots is polynomial (in the input size), we present a reduction to the  unsplittable flow problem on a path that adds only a constant factor to the approximation ratio.  We remark that the unsplittable flow problem considers only real-valued demands which is indeed simpler than our setting. Finally, we present a practical greedy algorithm for the single time slot case with an approximation ratio of $\tfrac{1}{2}\cos \frac{\phi}{2}$ and a running time of ${O}(N\log N)$, where $N$ is the total number of complex-valued user demands, which can be implemented in real world power systems.

The paper is structured as follows. In Sec.~\ref{sec:related}, we briefly present the related works. In Sec.~\ref{sec:model}, we provide the problem
definitions and notations needed. Then we present algorithms for the case of a  constant number of time slots in Sec.~\ref{sec:const}, namely, a PTAS for $\phi \in [0,\tfrac{\pi}{2}]$ and an FPTAS for $\phi \in [0,\pi \mbox{-} \varepsilon]$. In Sec.~\ref{sec:poly} we present the reduction to the unsplittable flow problem for the case of a polynomial number of time slots. The proposed greedy algorithm is provided in Sec.~\ref{sec:greedy}. In Sec.~\ref{sec:reduction}, we show how to include elastic demands, i.e., demands that can be partially satisfied,  along with the inelastic ones in the problem formulation. Lastly, Sec.~\ref{sec:concl} concludes this article. 
\section{Related work}\label{sec:related}

Several recent studies consider resource allocation with inelastic demands (that is, when the decision variables are all binary).
For a single time slot case, the problem studied here resembles the complex-demand knapsack problem {\sc (CKP)}~\cite{YC13CKP}. Let  $\phi$ be the maximum angle between any pair of complex-valued demands and $N$ be the total number of these demands. A $\frac{1}{2}$-approximation was obtained \cite{YC13CKP} for the case where $0 \le \phi \le \frac{\pi}{2}$.  On the other hand, it was shown in \cite{woeginger2000does} (also \cite{YC13CKP}) that no fully polynomial-time approximation scheme (FPTAS) exists. Recently, a  polynomial-time approximation scheme (PTAS), and a bi-criteria FPTAS (allowing constraint violation) for $\frac{\pi}{2} < \phi < \pi - \varepsilon$ were obtained in \cite{CKM14, CKM15}. This essentially closes the approximation gap as it is shown in \cite{KCE14} that when $\phi \in (\tfrac{\pi}{2}, \pi]$, there is no $\alpha$-approximation to {\sc CKP} for any $\alpha$ with polynomial number of bits, unless P=NP. Additionally, when $\varepsilon$ is arbitrarily close to zero (i.e., $\phi \to \pi$) there is no $(\alpha,\beta)$-approximation in general for any $\alpha,\beta$ with polynomial number of bits, unless P=NP. Therefore, the PTAS and the bi-criteria FPTAS~\cite{CKM14} are the best approximation possible for {\sc CKP}. 
In \cite{KT15}, an extension of {\sc CKP} was provided to handle a constant number of  quadratic (and linear) constraints. A fast greedy algorithm was given in \cite{KKCMZ16} for solving {\sc CKP} with a constant approximation ratio that runs in $O(N\log N)$ time. 
A recent work \cite{MCK16} extends the greedy algorithm to solve the optimal power flow problem (OPF) with inelastic demands,  a generalization of {\sc CKP} to a networked setting including voltage constraints.

When the demands are real-valued, the problem under study (considering multiple time slots)  is related to the unsplittable flow problem on a path ({\sc UFP}). In {\sc UFP}, each demand is associated with a unique path from a source to a sink.
{\sc UFP} is strongly NP-hard~\cite{darmann2010resource}.
A Quasi-PTAS was obtained by Bansal et al.~\cite{BNC07}.
Anagnostopoulos et al.~\cite{anagnostopoulos2014mazing} obtained a $1/(2+\epsilon)$-approximation (where $\epsilon>0$ is a constant). This matched the previously known approximation under the {\it no bottleneck assumption (NBA)} \cite{chekuri2007multicommodity}, which is the case when the largest demand is at most the smallest capacity.
The {\sc UFP} with bag constraints (\mbox{\sc bag-UFP}) is the generalization of {\sc UFP} where each user has a set of demands among which at most one is selected~\cite{chakaravarthy2010varying}. This problem is APX-hard even in the case of unit demands and capacities~\cite{spieksma1999approximability}. Under the NBA assumption, a $\tfrac{1}{65}$-approximation  was obtained  in ~\cite{elbassioni2012approximation}, which was later improved by \cite{chakaravarthy2010varying} to $\tfrac{1}{17}$. More recently, an $O(\log N/\log \log N)^{-1}$-approximation without NBA was obtained in \cite{Grandoni2015}. A constant factor approximation to {\sc bag-UFP} remains an interesting open question.

In this paper, we extend the complex-demand knapsack problem over a discretized time horizon, where each time slot is associated with a fixed supply limit. A user provides multiple demand preferences with their respective time window from which at most one is selected. When the number of time slots is constant, the problem generalizes {\sc CKP} (see,~\cite{CKM14}) to multiple time slots, and also extends that of~\cite{KT15} by considering multiple demands per user, thereby adding $n$ extra constraints, where $n$ is the number of users. 
Furthermore, for the case of a polynomial number of time slots, our problem is a generalization of the unsplittable flow problem on paths to accommodate complex-valued demands. Finally, we extend the greedy algorithm in~\cite{KKCMZ16} (for the single time slot case) to handle multiple demands per user keeping the same approximation ratio and running time.

\section{Problem Definitions and Notations}\label{sec:model}
In this section we formally define the complex-demand scheduling  problem.
Throughout this paper, we sometimes denote $\nu^{\rm R} \triangleq {\rm Re}(\nu)$ as the real part and $\nu^{\rm I} \triangleq {\rm Im}(\nu)$ as the imaginary part of a given complex number $\nu\in \CC$. We use $|\nu|$ to denote the magnitude of $\nu$ and ${\rm arg}(\nu)$ to denote the angle $\nu$ makes with the positive real axis. Unless stated otherwise, we denote $\mu_t$ (and sometimes $\mu(t)$ whenever we use subscripts for other purposes) as the $t$-th component of the vector  $\mu$.
\subsection{Complex-demand Scheduling  Problem }

Consider a discrete time horizon denoted by $\cT\triangleq\{1,...,m\}$. At each time slot $t\in \cT$, the generation capacity of the power grid is denoted by $C_t\in \RR_+$. Denote by $\cN\triangleq\{1,...,n\}$ the set of all users with cardinality $n \triangleq |\cN|$. Each user $k\in \cN$ declares a set of demand preferences indexed by the set $D_k$. Each demand $j \in D_k$ is defined over a time interval $T_{j} \subseteq \cT$, that is, $T_j = \{t_1,t_1+1,...,t_2\}$ where $t_1,t_2\in \cT$ and $t_1\le t_2$.  Demand $j$ is also associated with a set of complex numbers $\{s_{k,j}(t)\}_{ t\in T_{j}}$ where $s_{k,j}(t) \triangleq s_{k,j}^{\rm R}(t) + {\bf i} s_{k,j}^{\rm I}(t) \in \CC$ is a complex power demand at time $t$. A positive utility $u_{k,j}$ is associated with each user demand $(k,j)$ if  satisfied.

The goal is to find a solution of control variables $(x_{k,j})_{k \in \cN, j\in D_k} \in \{0,1\}^{\sum_{i \in \cN}|D_i|}$ that maximizes the total utility of satisfied users subject to the generation capacity over time. 
More formally, we define the complex-demand scheduling problem over $m$ discrete time slots (\textsc{$m$-CSP}) by the following integer programming problem.
\begin{align}
\textsc{($m$-CSP)} \quad& 
\displaystyle \max  \sum_{k\in \cN } \sum_{j \in D_k} u_{k,j} x_{k,j} \\
\text{subject to }  \quad  &  \displaystyle\Big|\sum_{k\in \cN} \sum_{j\in D_k: T_{j} \ni t} s_{k,j}(t)\cdot x_{k,j}\Big| \le C_t, \qquad \text{ for all } t \in \cT \label{mc1}\\
&\sum_{j \in D_k} x_{k,j} \le 1, \qquad \text{ for all } k \in \cN \label{mc2}\\
& x_{k,j}\in\{0,1\},\qquad \text{ for all }(k,j)\in\cI, \label{mc3} 
\end{align}
where $\cI=\{(k,j):~k\in\cN~\text{ and }~j\in D_k\}$.
Cons.~\raf{mc1} captures the capacity limit, and Cons.~\raf{mc2} forces at most one demand for every user to be selected. Note that \raf{mc1}  is equivalent to a quadratic constraint 
{\small
$$
\left(\sum_{k\in \cN} \sum_{j\in D_k: T_{j} \ni t} \re(s_{k,j}(t))\cdot x_{k,j}\right)^2+\left(\sum_{k\in \cN} \sum_{j\in D_k: T_{j} \ni t} \im(s_{k,j}(t))\cdot x_{k,j}\right)^2 \le C_t^2, ~ \forall  t \in \cT.
$$}
We consider the following assumption that are mainly needed in Sec.~\ref{sec:bfptas}: for any user $k$, \vspace{-5pt}
\begin{itemize}
	\item all demands $s_{k,j}(t)$, for $j \in D_k$ and $t\in T_j$, reside in the same quadrant of the complex plane.
\end{itemize}
We also assume without loss of generality that $u_{k,j}>0$ and $|s_{k,j}(t)|\le C_t$ for all $(k,j)\in \cI$ and $t\in \cT$. 
Problem {\sc $1$-CSP} (i.e., $|\cT|=1$) is called the complex-demand knapsack, denoted by {\sc CKP}. Evidently, {\sc $m$-CSP} is NP-complete, since the knapsack problem is its special case when we set all $s_{k,j}^{\rm I}(1) = 0$, $\cT = \{1\}$, and $|D_k|=1$. We will write {\sc $m$-CSP}$[\phi_1,\phi_2]$ for the restriction of problem {\sc $m$-CSP} subject to $\phi_1 \le \max_{k \in {\cal N}}{\rm arg}(s_{k,j}(t)) \le \phi_2$, where we assume ${\rm arg}(s_{k,j}(t))\ge 0$ for all $(k,j)\in\cI$, $t\in T_j$. 

\subsection{Approximation Algorithms}
Given a solution $x \triangleq (x_{k,j})_{k \in \cN, j \in D_k}$, denote the total utility by $u(x)\triangleq \sum_{k \in \cN}$ $\sum_{j \in D_k} u_{k,j} x_{k,j}$. We denote an optimal solution to \textsc{$m$-CSP} by $x^\ast$ and $\OPT \triangleq u(x^\ast)$. With a slight abuse of notation, for a given subset $S\subseteq \cN$, we write $u(S)\triangleq \sum_{k\in S}\sum_{j\in D_k} u_{k,j}$.

\begin{definition}
	For $\alpha\in(0,1]$ and $\beta\ge1$, we define a bi-criteria $(\alpha,\beta)$-approximation to \textsc{$m$-CSP} as a solution $\hat x =(\hat{x}_{k,j})_{(k,j) \in \cI} \in \{0, 1\}^{|\cI|}$ satisfying Cons.~\raf{mc2}-
	\raf{mc3}, and
	\begin{align}
	&   \displaystyle\Big|\sum_{k\in \cN} \sum_{j\in D_k: T_{j} \ni t } s_{k,j}(t) \hat x_{k,j}\Big| \le \beta \cdot  C_t \qquad \text{ for all } t \in \cT \label{eq:betamc1}
	\end{align}
	such that $u(\hat x) \ge \alpha \OPT$.
\end{definition} 
In the above definition, $\alpha$ characterizes the approximation ratio between an approximate solution and the optimal solution, whereas $\beta$ characterizes the violation bound of constraints.
In particular, {\em polynomial-time approximation scheme} (PTAS) is a $(1-\epsilon,1)$- approximation algorithm for any $\epsilon>0$.  The running time of a PTAS is polynomial in the input size for every fixed $\epsilon$, but the exponent of the polynomial might depend on $1/\epsilon$.  
An even stronger notion is a {\em fully polynomial-time approximation scheme} (FPTAS), which requires the running time to be polynomial both in input size and $1/\epsilon$. In this paper, we are interested in bi-criteria FPTAS, which is a $(1, 1+\epsilon)$-approximation algorithm for any $\epsilon>0$, with the running time to be polynomial in the input size and  $1/\epsilon$.  When $\beta=1$, we sometimes  call an $(\alpha,\beta)$-approximation  an ${\alpha}$-approximation.



 
\section{{$m$-CSP} with a Constant Number of Time Slots} \label{sec:const}
In this section we assume the number of time slots $|\cT|$ is a constant.
This assumption is practical in the  realistic setting, where users declare their demands on hourly basis one day ahead in the electricity market.
We remark that the results in this and the next section do not require $T_j$ to be a continuous interval in $\cT$.

\subsection{PTAS for {\sc $m$-CSP$[0,\frac{\pi}{2}]$}}
\label{sec:cdrconst}
Define a convex relaxation of {\sc $m$-CSP} (denoted by {\sc rlxCSP}), such that Cons.~\raf{mc3} are replaced by $x_{k,j} \in [0,1]$ for all $(k,j)\in \cI$.
We define another convex relaxation that will be used in the PTAS denoted by {\sc rlxCSP$[S_1,S_0]$} which is equivalent to {\sc rlxCSP}, subject to partial substitution such that $x_{k,j} = 1$, for all ${(k,j)} \in S_1$ and $x_{k,j} = 0,$ for all $(k,j)\in S_0$, where $S_1,S_0\subseteq\cI$ such that $S_1 \cap  S_0 = \varnothing$:
\begin{align}
&\textsc{(rlxCSP}[S_1, S_0]\text{)} \qquad \displaystyle  \max_{x_{k,j} \in [0,1]} \,  \sum_{k\in \cN }\sum_{j \in D_k} u_{k,j} x_{k,j}, ~~\text{s.t.}  \\
& \displaystyle \Big(\sum_{k \in \cN}\sum_{j\in D_k: t \in T_{j}}s_{k,j}^{\rm R}(t)\cdot x_{k,j}\Big)^2+ \Big(\sum_{k\in \cN}\sum_{j \in D_k: t \in T_{j} }s_{k,j}^{\rm I}(t)\cdot x_{k,j}\Big)^2\le C_t^2, ~\forall t\in\cT \label{CV1}\\
&\sum_{j \in D_k} x_{k,j} \le 1, \qquad \text{ for all } k \in \cN \\
&x_{k,j} = 1, \qquad \text{ for all } (k,j) \in S_1\\
&x_{k,j} = 0\qquad \text{ for all } (k,j) \in S_0.
\end{align}

The above relaxation can be solved approximately in polynomial time using standard convex programming algorithms (see, e.g., \cite{nemirovski2008interior}). In fact, such algorithms can find a feasible solution $\xcp$ to the convex relaxation such that $u(\xcp) \ge \OPT^* - \delta$, in time polynomial in the input size (including the bit complexity) and $\log\tfrac{1}{\delta}$, where $\OPT^*$ is the optimal objective value of {\sc rlxCSP$[S_1, S_0]$}. Notice that the value of an optimal solution to $\textsc{rlxCSP}[S_1, S_0]$ problem is no worse than that of {\sc $m$-CSP} since the feasibility region of the latter is a subset of that of the former. This, in turn, implies that $\OPT^* \ge \OPT\ge \bar u \triangleq \max_{k,j} u_{k,j}$, and hence setting $\delta$ to $\tfrac{\epsilon}{2} \cdot \bar u $ assures that $u(\xcp) \ge (1-\frac{\epsilon}{2})\cdot \OPT^*$.

 We provide a $(1-\epsilon,1)$-approximation  for  \textsc{$m$-CSP$[0,\frac{\pi}{2}]$} in Algorithm~\ref{alg:mcptas}, denoted by  {\sc $m$-CSP-PTAS}. The idea of   {\sc $m$-CSP-PTAS}  is based on that proposed in \cite{KT15,EN17} with two extensions. First, we consider multiple demands per user. This in fact  adds $n$ extra  constraints to that in \cite{KT15,EN17}, and thus the rounding procedure requires further analysis. The second extension is the addition of  elastic demands $\cF$. We remark that \cite{CKM14} considers multiple inelastic demands per user for the single time slot case (denoted by {\sc CKP}); however, their algorithm is based on a completely different geometric approach that is more complicated than that in~\cite{KT15}.

Given a feasible solution $x^*$ to {\sc rlxCSP$[S_1,S_0]$}, a restricted set of demands $S\subseteq\cI\cup \cF$, and vectors $c^1,c^2 \in \RR_+^m$, we define the following relaxation, denoted by {\sc LP}$[c^1,c^2, x^*, S]$:
\begin{align}
\textsc{(LP$[c^1,c^2, x^*, S]$)}&\quad  \displaystyle  \max_{x_{k,j} \in [0,1]} \, \sum_{k\in \cN}\sum_{j \in D_k} u_{k, j} x_{k,j}\\
\text{ s.t }\quad & \displaystyle  \sum_{k\in \cN}\sum_{j \in D_k: t \in T_{j}} s^{\rm R}_{k,j}(t)\cdot x_{k,j} \le c^1_t,\quad \text{for all } t \in \cT \label{md1}\\
& \displaystyle \sum_{k\in \cN}\sum_{j \in D_k: t \in T_{j}} s^{\rm I}_{k,j}(t)\cdot x_{k,j}\le c^2_t,\quad \text{for all } t \in \cT\label{md2}\\
&\sum_{j \in D_k} x_{k,j} \le 1, \qquad \text{ for all } k \in \cN \label{md3} \\
&x_{k,j} = x^*_{k,j} \qquad \text{ for all } (k,j) \in S\label{md4}. 
\end{align}

 The Algorithm~\ref{alg:mcptas} proceeds as follows. We guess $S_1 \subseteq \cI$ to be the set  of largest-utility $\frac{8m}{\epsilon}$ inelastic demands in the optimal solution; this defines an excluded set of demands $S_0 \subseteq\cI\setminus S_1$ whose utilities exceed one of the utilities in $S_1$ (Step~\ref{alg:mcv}). For each such $S_1$ and $S_0$, we solve the convex program {\sc rlxCSP$[S_1,S_0]$} and obtain a $(1-\tfrac{\epsilon}{2})$-approximation $\xcp$ (note that the feasibility of the convex program is guaranteed by the conditions in Step~\ref{alg:mc-feas}). The real and imaginary projections over all time slots of solution $\xcp$, denoted by $L^{\rm R} \in \RR^m_+$ and $L^{\rm I} \in \RR^m_+$, are used to define the linear program {\sc LP}$[L^{\rm R}, L^{\rm I}, \xcp, S_1 \cup S_0]$ over the restricted set of demands $S_1 \cup S_0$.
  We  solve  the linear program in Step \ref{alg:mclp}, and then round down the solution corresponding to demands $(k,j) \in \cI$ in Step \ref{alg:mc-round}.  Finally, we return a solution $\hat x$ that attains maximum utility among all obtained solutions.

\begin{algorithm}[!htb]
	\caption{ {\sc $m$-CSP-PTAS}$[\{u_{k,j}, \{s_{k,j}(t)\}_{t\in T_{j}}\}_{k\in\cN, j\in D_k},(C_t)_{t\in \cT},\epsilon]$} \label{alg:mcptas}
	\begin{algorithmic}[1]
	\Require Users' utilities and demands $\{u_{k,j},\{s_{k,j}(t)\}_{t\in T_{j}}\}_{k\in\cN, j\in D_k}$; capacity over time $C_t$; accuracy parameter $\epsilon$
		\Ensure $(1-\epsilon,1)$-solution $\hat{x}$ to \textsc{$m$-CSP$[0,\frac{\pi}{2}]$} 
		\State $\hat x \leftarrow { \bf 0}$
		\For {each set  $S_1\subseteq \cI$ such that $|S_1| \le \frac{8m}{\epsilon}$ } \label{alg:guess}
		\If{$\Big|\displaystyle\sum_{(k,j)\in S_1: t \in T_{j}}s_{k,j}(t)\Big|\le C_t$ for all $t\in \cT$ and $\displaystyle |\{j:(k,j)\in S_1\}| \le 1, \text{ for all } k \in \cN $} \label{alg:mc-feas}
		\State $S_0\leftarrow\{(k,j)\in\cI\setminus S_1 \mid u_{k,j}>\min_{(k',j')\in S_1}u_{k',j'}\}$ \label{alg:mcv}
		\State $\xcp \leftarrow$ Solution of {\sc rlxCSP$[S_1,S_0]$}\label{alg:cvx} \Comment{Obtain a $(1-\tfrac{\epsilon}{2})$-approximation }
		\For {all $t\in \cT$ }
		\State  $\displaystyle L_t^{\rm R} \leftarrow \sum_{k\in \cN}\sum_{j \in D_k: t \in T_{j}} s_{k,j}^{\rm R}(t)\cdot \xcp_{k,j} $; $\displaystyle L_t^{\rm I} \leftarrow \sum_{k\in \cN}\sum_{j \in D_k: t \in T_{j}} s_{k,j}^{\rm I}(t)\cdot\xcp_{k,j} $  
		\EndFor
		\State $\xlp \leftarrow$ Solution of {\sc LP}$[L^{\rm R}, L^{\rm I}, \xcp,S_1 \cup S_0]$ \label{alg:mclp}
		 \Statex \Comment{Round the LP solution} 
		 \State \mbox{$\bar x \leftarrow \{(\bar x_{k,j})_{k\in\cN,j\in D_k} \mid \bar x_{k,j} = \lfloor\xlp_{k,j}\rfloor \text{ for } (k,j) \in \cI\}$}\label{alg:mc-round}
		 \If{$u(\bar x) > u(\hat x)$ }
		 \State $\hat x \leftarrow \bar x$
		 \EndIf 
		 \EndIf
		 \EndFor
		\State \Return $\hat x$
	\end{algorithmic}
\end{algorithm}
\begin{theorem}
	\label{th:mcptas}
	For any fixed $\epsilon$, Algorithm~\ref{alg:mcptas}  obtains a $(1- \epsilon,1)$-approximation in polynomial time. 
\end{theorem}

We remark  that a PTAS is the best approximation one can hope for, since it is shown in \cite{YC13CKP,woeginger2000does} that it is NP-Hard to obtain an FPTAS for the single time slot version  ({\sc $1$-CSP$[0,\tfrac{\pi}{2}]$}). 

\begin{proof}
	One can easily see that the running time of Algorithm~\ref{alg:mcptas} is polynomial in size of the input, for any given $\epsilon$. We now argue that the solution $\hat x$ is $(1- \epsilon)$-approximation for {\sc $m$-CSP$[0,\tfrac{\pi}{2}]$}. Let $x^*$ be the optimal solution for {\sc $m$-CSP$[0,\tfrac{\pi}{2}]$} of utility $\OPT \triangleq u(x^*)$. Define  $S^*\triangleq\{(k,j) \in \cI \mid x^*_{k,j} = 1\}$. 
	By the feasibility of $x^*$, in Step~\ref{alg:cvx} the algorithm obtains  
	\begin{equation}
	u(\xcp) \ge (1-\tfrac{\epsilon}{2})\cdot \OPT^* \ge (1-\tfrac{\epsilon}{2})\cdot \OPT, \label{eq:mc-cvx}
	\end{equation}
	where $\OPT^*$ is the optimal value of {\sc rlxCSP$[S_1,S_0]$} for some $S_1$ equal to the highest $\tfrac{8m}{\epsilon}$ utility demands in $S^*$, and $S_0\cap S^* = \varnothing$.
	If $|S^*| \le \tfrac{8m}{\epsilon}$,  then obviously $\hat x = \xlp = \xcp$ and $u(\xcp) \ge  (1-\tfrac{\epsilon}{2}) \OPT $.
	
	Now suppose $|S^*|> \tfrac{8m}{\epsilon}$. Observe that $\xcp$ is feasible for   {\sc LP}$[L^{\rm R}, L^{\rm I}, \xcp, S_1 \cup S_0]$ (Cons.~\raf{md1}-\raf{md4} are tight when $\xcp$ is substituted). Therefore, the optimal solution $\xlp$ of {\sc LP}$[L^{\rm R}, L^{\rm I}, \xcp, S_1 \cup S_0]$ satisfies 
	\begin{align}
	u(\xlp) \ge u(\xcp) \label{eq:mc-lp}.
	\end{align}
	
	By Lemma~\ref{lem:2m} below,  {\sc LP}$[L^{\rm R}, L^{\rm I}, \xcp, S_1 \cup S_0]$ has a basic feasible solution (BFS) with at most $4m$ fractional components, and for any fractional component $(k,j)\in\cI\setminus (S_1\cup S_0)$, $u_{k,j} < \min_{(k',j') \in S_1} u_{k',j'} \le \frac{\sum_{(k',j') \in S_1} u_{k',j'}}{|S_1|}$. Therefore, rounding down $\xlp$ in Step~\ref{alg:mc-round} gives,
	\begin{align*}
	u(\hat x) &\ge u(\xlp) - 4m \frac{\sum_{(k,j)\in S_1} u_{k,j}}{|S_1|} \ge (1-\tfrac{\epsilon}{2}) u(\xlp) \\
	& \ge (1-\tfrac{\epsilon}{2})^2 \cdot \OPT \ge (1-\epsilon)\cdot \OPT,
	\end{align*}
	where the second to last inequalities follow by Eqns.~\raf{eq:mc-cvx}-\raf{eq:mc-lp}. It remains to show that $\hat x$ is feasible. Since $\hat x$ is obtained by rounding down (some) $\xlp$ (Step.~\ref{alg:mc-round}),
{\small
	\begin{align}
	&\Big(\sum_{k\in \cN}\sum_{j\in D_k: t \in T_{j}}s_{k,j}^{\rm R}(t)\cdot \hat x_{k,j}\Big)^2+ \Big(\sum_{k\in \cN}\sum_{j \in D_k: t \in T_{j} }s_{k,j}^{\rm I}(t)\cdot \hat x_{k,j}\Big)^2 \\
	&\le \Big(\sum_{k\in \cN}\sum_{j\in D_k: t \in T_{j}}s_{k,j}^{\rm R}(t)\cdot \xlp_{k,j}\Big)^2+ \Big(\sum_{k\in \cN}\sum_{j\in D_k: t \in T_{j}}s_{k,j}^{\rm I}(t)\cdot \xlp_{k,j}\Big)^2 \notag\\
	&\le (L_t^{\rm R})^2 + (L_t^{\rm I})^2 = \Big(\sum_{k\in \cN}\sum_{j\in D_k: t \in T_{j}} s_{k,j}^{\rm R}(t) \xcp_{k,j}\Big)^2 + \Big(\sum_{k\in \cN}\sum_{j\in D_k: t \in T_{j}} s_{k,j}^{\rm I}(t) \xcp_{k,j}\Big)^2 \le C^2_t \label{eq:mc-feas},
	\end{align}}
	where Eqn.~\raf{eq:mc-feas} follows by the feasibility of $\xlp$ and $\xcp$ respectively. Hence, Cons.~\raf{mc1} are satisfied. Finally, since some components of $\xlp$ in Step~\ref{alg:mc-round} are only rounded down, Cons.~\raf{mc2}-\raf{mc3} are also satisfied. 
\end{proof}

\begin{lemma}[\cite{SR10}]{\label{lem:2m}}
	Let $x$ be a basic feasible solution (BFS) for {\sc LP}$[c^1,c^2, x^*, S]$. Then $x$ has at most $4m$ non-integral components. 
\end{lemma}
\begin{remark}
	The above proof shows that we do not need actually to solve {\sc LP}$[L^{\rm R}, $ $L^{\rm I}, \xcp, S_1 \cup S_0]$; starting from $\xcp$, we only need to get a BFS with the same (or better) objective value, which can be reduced to solving systems of linear equations. 
\end{remark}
\begin{proof}
	Let $h$ be the number of users $k$ such that $\sum_{j\in D_k}x_{k,j}=1$.  
	By the properties of a BFS (see, e.g., \cite{GLS88,S86}), the number of strictly positive components in $x$ is at most $2m+h$. Furthermore, constraints \raf{md3} impose that for each $k\in\cN$ among those $h$ users, there is a $j\in D_k$ such that $x_{k,j}>0$. The remaining $2m$ positive variables can belong to at most $2m$ of the constraints \raf{md3}, implying that at least $\max\{h-2m,0\}$ variables are set to $1$. It follows that the total number variables taking non-integral values is at most $2m+h-\max\{h-2m,0\}\le 4m$.  
\end{proof}

\subsection{Bi-criteria FPTAS for \textsc{$m$-CSP$[0,\pi\mbox{-}\varepsilon]$}} \label{sec:bfptas} 
In the previous section, we have restricted our attention to the setting where all demands lie in the positive quadrant of the complex plane (i.e., \textsc{$m$-CSP$[0,\frac{\pi}{2}]$}). In this section, we extend this setting to the second quadrant (\textsc{$m$-CSP$[0,\pi\mbox{-}\varepsilon]$}) for any arbitrary small constant $\varepsilon > 0$, that is, we assume ${\rm arg}(s_{k,j}(t)) \le \pi-\varepsilon$ for all $k \in \cN, j \in D_k, t \in T_j$. 
It is shown in \cite{KCE14} that for \textsc{$1$-CSP$[0, \pi]$}  (the case $|\cT|=1$) there is no $(\alpha, 1)$-approximation for \textsc{$1$-CSP$[0,\pi \mbox{-} \varepsilon]$} unless P=NP. Therefore, a bi-criteria $(1,1+\epsilon)$ is the best approximation one can hope for. Additionally, it is shown that if $\varepsilon$ is arbitrarily close to zero, then there is no $(\alpha,\beta)$-approximation in general for any $\alpha,\beta$ with polynomial number of bits, unless P=NP. Thus, one should expect the running time of $(1,1+\epsilon)$ to depend on the maximum angle $\phi\triangleq\max_{k\in \cN,j\in D_k, t\in T_j}{\rm arg}(s_{k,j}(t))$. We present below such an algorithm, which is an extension of that presented by \cite{CKM14} for multiple time slots.

For convenience, we let $\theta = \max\{\phi - \frac{\pi}{2},0\}$
(see Fig.~\ref{fig:fptas} for an illustration).
We present a  $(1,1+\epsilon)$-approximation  for \textsc{$m$-CSP$[0,\pi\mbox{-}\varepsilon]$} in Algorithm~\ref{CSP-biFPTAS}, denoted by  \mbox{\sc $m$-CSP-biFPTAS}, with running time polynomial in both $\frac{1}{\epsilon}$ and $n$ (i.e., FPTAS). We assume that $\tan \theta$ is bounded  by a polynomial in $n$; as mentioned above, without this assumption, a bi-criteria FPTAS is unlikely to exist (see \cite{KCE14}). 
\begin{figure}[!htb]
	\begin{center}
		\includegraphics[scale=.7]{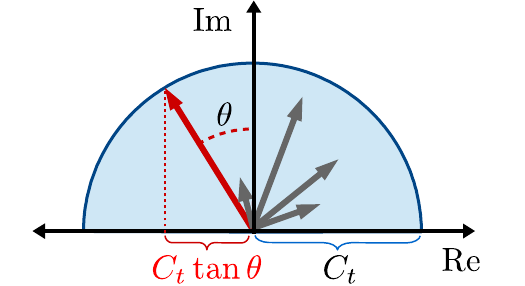}
	\end{center}\vspace{-10pt}
	\caption{ We measure $\theta = \phi - \frac{\pi}{2}$ from the imaginary axis.}
	\label{fig:fptas}
\end{figure} 

Let $\cN_+\triangleq \{k \in \cN \mid s^{\rm R}_{k,j}(t) \ge 0, \forall j\in D_k,~t\in T_j\}$ and $\cN_-\triangleq\{k \in \cN \mid s^{\rm R}_{k,j}(t) < 0,\forall j\in D_k,~t\in T_j\}$ be the subsets of users with demands in the first and second quadrants, respectively. Note that $\cN_+$ and $\cN_-$ partition the set of users $\cN$ by the assumption stated in Sec.~\ref{sec:model}.
 Consider any  solution $\widehat x$ to \textsc{$m$-CSP$[0,\pi\mbox{-}\varepsilon]$}.
The basic idea of Algorithm {\sc $m$-CSP-biFPTAS} is to enumerate the guessed total projections on real and imaginary axes of all time slots for $\sum_{k \in \cN_+} \sum_{j \in D_k: t \in T_j} \widehat x_{k,j}  s_{k,j}(t)$ and $\sum_{k \in \cN_-} \sum_{j \in D_k: t \in T_j} \widehat x_{k,j}  s_{k,j}(t)$, respectively. 
We can use $\tan \theta$ to upper bound the total projections for any feasible solution $\widehat x$ (see Fig.~\ref{fig:fptas} for a pictorial illustration) as follows, for all $t$:
\begin{gather}
\sum_{k \in \cN} \sum_{j\in D_k: t \in T_j} s^{\rm I}_{k,j}(t)\cdot \widehat x_{k,j} \le C_t, \quad
\sum_{k \in \cN_- }\sum_{j \in D_k: t \in T_j} - s^{\rm R}_{k,j}(t)\cdot \widehat x_{k,j}    \le  C_t \tan \theta, \notag\\ 
\sum_{k \in \cN_+}\sum_{j \in D_k: t \in T_j}  s^{\rm R}_{k,j}(t) \cdot \widehat x_{k,j}   \le C_t(1+ \tan \theta). \label{eq:ubounds}
\end{gather}

We then solve two separate multi-dimensional knapsack problems of dimension $2m$  (denoted by {\sc $2m$DKP}), to find subsets of demands that satisfy the individual guessed total projections. But since {\sc $2m$DKP} is generally NP-hard, we need to round-up the demands to get a problem that can be solved efficiently by dynamic programming. We show that the violation of the optimal solution to the rounded problem w.r.t. to the original problem is small in $\epsilon$. 

Next, we describe the rounding in detail. 
 First, define
$L_t \triangleq \frac{\epsilon C_t}{n (\tan \theta+1)}, \text{ for all }$ $ t \in \cT$
 such that the new rounded demands $\widehat s_{k,j}(t)$ are defined by:
\begin{equation}\hspace{-5pt}
\widehat s_{k,j}(t) =
\widehat s^{\rm R}_{k,j}(t) + {\bf i} \widehat s^{\rm I}_{k,j}(t) \triangleq 
\left\{\begin{array}{ll}
\left\lceil \frac{s^{\rm R}_{k,j}(t)}{L_t} \right\rceil \cdot L_t + {\bf i} \left\lceil \frac{s^{\rm I}_{k,j}(t)}{L_t} \right\rceil \cdot L_t, &\text{ if }s^{\rm R}_{k,j}(t)\ge 0,\\[3mm]
\left\lfloor \frac{s^{\rm R}_{k,j}(t)}{L_t} \right\rfloor \cdot L_t + {\bf i} \left\lceil \frac{s^{\rm I}_{k,j}(t)}{L_t} \right\rceil \cdot L_t, & \text{ otherwise. } 
\end{array}\right.
\label{eq:truc}
\end{equation}
 
For convenience, we assume that $s_{k,j}(t)=0$ if $t\in\cT\setminus T_j$. Let $\xi_+ \in \RR^m_+$ (and $\xi_- \in \RR^m_+$), $\zeta_+ \in \RR^m_+$ (and $\zeta_- \in \RR^m_+$) be respectively the guessed real and imaginary absolute total rounded projections of an optimal solution. Then, the possible values of $\xi_+, \xi_-, \zeta_+$ and $\zeta_-$ in each component $t\in\cT$ are integer mutiples of $L_t$:
\begin{align}
 \xi_+(t) \in {\cal A}_+(t)&\triangleq \left\{0, L_t, 2 L_t,\ldots,\left\lceil \frac{C_t (1 + \tan \theta )}{L_t} \right\rceil \cdot L_t\right\}, \notag\\ 
 \xi_-(t) \in {\cal A}_-(t)&\triangleq \left\{0,L_t, 2L_t,\ldots, \left\lceil \frac{C_t \cdot \tan \theta }{L_t} \right\rceil\cdot L_t \right\},\notag\\
 \zeta_+(t),\zeta_-(t)  \in {\cal B}(t)  &\triangleq \left\{0,  L_t, 2L_t,\ldots,\left\lceil \frac{C_t}{L_t} \right\rceil\cdot L_t\right\}.
\label{eq:grid}
\end{align}

The next step is to solve the rounded instance exactly. Assume an arbitrary order on $\cN = \{ 1, ..., n\}$. We use recursion to define a table, with each entry ${U}(k,c^1,c^2)$, $c^1,c^2\in \RR^{m}_+$, as the maximum utility obtained from a subset of users $\{1,2,\dots,K\} \subseteq \cN$ with demands $\{\widehat s_{k,j}(t)\}_{k\in\{1,...,K\}, j \in D_k, t \in T_j}$  that can fit exactly (i.e., satisfy the capacity constraints with equality) within capacities $\{c^1_t\}_{t=1,...,m}$ on the real axis and $\{c^2_t\}_{t= 1,...,m}$ on the imaginary axis. 
We denote by {\sc $2m$DKP-Exact}$[\cdot]$ the algorithm for solving exactly the rounded {\sc $2m$DKP} by  dynamic programming. We provide the detailed description of {\sc $2m$DKP-Exact}$[\cdot]$ in Algorithm~\ref{alg:dyn}.


\begin{algorithm}[!htb]
	\caption{{\sc $m$-CSP-biFPTAS}$[\{u_{k,j}, \{s_{k,j}(t)\}_{t\in T_{j}}\}_{k\in\cN, j\in D_k},(C_t)_{t\in \cT},\epsilon]$}\label{CSP-biFPTAS}
	\begin{algorithmic}[1]
		\Require Users' utilities and demands  $\{u_{k,j},\{s_{k,j}(t)\}_{t\in T_{j}}\}_{k\in\cN, j\in D_k}$; capacity over time $C_t$; accuracy parameter $\epsilon$
		\Ensure $(1,1+4\epsilon)$-solution $\widehat{x}$ to \textsc{$m$-CSP$[0,\pi\mbox{-}\varepsilon]$}
		\State $\widehat{x} \leftarrow \bzero$
		\ForAll {$s_{k,j}(t)$, $k \in \cN$, $j\in D_k$, and $t\in T_j$}
		\State Set $\widehat s_{k,j}(t) \leftarrow \widehat s_{k,j}^{\rm R}(t) + {\bf i} \widehat s_{k,j}^{\rm I}(t) $ as defined by \raf{eq:truc}
		\EndFor
		\ForAll {$\xi_+ \in  \prod_{t\in \cT} {\cal A}_+(t), \xi_- \in  \prod_{t\in \cT}{\cal A}_-(t), \zeta_+, \zeta_- \in \prod_{t\in \cT}{\cal B}(t)$}
		\If {$\big(\xi_+(t) - \xi_-(t)\big)^2 + \big(\zeta_+(t) + \zeta_-(t)\big)^2 \le (1+2\epsilon)^2C_t^2$ for all $t\in\cT$}\label{cond1}
		\State {\small $  y_+ \leftarrow \text{\sc $2m$DKP-Exact}\Big(\{u_{k,j},(\widehat s_{k,j}(t)/L_t)_t\}_{k\in\cN_+, j \in D_k}, \big(\xi_+(t)/L_t\big)_t,\big(\zeta_+(t)/L_t\big)_t\Big)$}
		\State {\small $  y_- \leftarrow \text{\sc $2m$DKP-Exact}\Big(\{u_{k,j},(-\widehat s_{k,j}(t)/L_t)_t\}_{k\in \cN_-, j \in D_k}, \big(\xi_-(t)/L_t\big)_t,\big(\zeta_-(t)/L_t\big)_t \Big)$ }
		\If{$u(y_+ + y_-) > u (\widehat{x})$}
		\State $\widehat{x} \leftarrow  y_+ + y_-$
		\EndIf 
		\EndIf
		\EndFor
		\State \Return $\widehat{x}$
	\end{algorithmic}
\end{algorithm}

\begin{algorithm}[!htb]
	\caption{{\sc $2m$DKP-Exact}$[\{u_{k,j}, \{\widehat s_{k,j}(t)\}_{t\in T_{j}}\}_{k \in \cW, j \in D_k},(c^1_t)_{t\in \cT},(c^2_t)_{t\in \cT}]  $} 
	
	\label{alg:dyn}
	\begin{algorithmic}[1] 
		\Require Utilities, and rounded demands of a restricted set of users $\cW\subseteq \cN$,  $\{u_{k,j},\{s_{k,j}(t)\}_{t\in T_{j}}\}_{k \in \cN, j \in D_k}$; integer capacity  vectors $(c^1_t)_{t\in \cT}, (c^2_t)_{t\in \cT}$
		\Ensure A utility-maximizing optimal binary vector $y$ subject to the capacity constraints defined by $c^1_t,c^2_t$
		
		\State Create a table of size $|\cW| \cdot \prod_t(c^1_t+1)\cdot(c^2_t+1)$, with each entry ${U}(k,c^1, c^2)$ according to: \label{s1}
	{\small 	\begin{align*}
		U(1,c^1,c^2) &\triangleq 
		\displaystyle\max_{j\in D_1 }\{u_{1,j} \mid \widehat s_{1,j}^{\rm R}(t) = c^1_t,~~ \widehat s_{1,j}^{\rm I}(t) = c^2_t, ~\forall t   \} 
\\
		U(k,c^1,c^2) &\triangleq \max 	\left\{ \begin{array}{l}
		\displaystyle\max_{j\in D_k }\Big\{u_{k,j} + U\big(k-1, (c^1_t -  \widehat s_{k,j}(t))_t, (c^2_t -  \widehat s_{k,j}(t))_t\big) \Big\},
		U(k-1,c^1,c^2)\end{array} \right\}\\
		U(k,c^1,c^2) &\triangleq -\infty \text{ for all } c^1,c^2 \not\in\RR^m_+
		\end{align*}}
		\State Compute the corresponding binary vector ${y}(k,c^1, c^2)$ according to the computations in step~\ref{s1}   
		\State \Return $y(|\cW|,c^1,c^2)$.
	\end{algorithmic}
\end{algorithm}

\begin{theorem}\label{thm:bptas}
	Algorithm {\sc $m$-CSP-biFPTAS} is a $(1,1+4\epsilon)$-approximation for \textsc{$m$-CSP$[0,\pi\mbox{-}\varepsilon]$} and its running time is polynomial in both $n$, $\big|\bigcup_{k}D_k\big|$, and $\frac{1}{\epsilon}$.
\end{theorem}
\begin{proof}
	First, the running time is proportional to the number of guesses, upper bounded by $(\frac{1}{\epsilon} n (\tan \theta+1))^{O(m)}$. 
	For each guess, {\sc $2m$DKP-Exact} constructs a table of size at most $(\frac{1}{\epsilon} n (\tan \theta+1))^{O(m)}$. Since we assumed $\tan \theta$ is polynomial in $n$, the total running time is polynomial in $n$ and $\frac{1}{\epsilon}$, if $m=O(1)$.
	
	To show the approximation ratio of 1, we note that {\sc $m$-CSP-biFPTAS} enumerates over all possible rounded projections subject to the capacity constraints in {\sc $m$-CSP} and that {\sc $2m$DKP-Exact} returns the exact optimal solution for each rounded problem. In particular, by Lemma \ref{lem-trunc} below one of the choices would be the rounded projection for the optimum solution $x^\ast$.   
	It remains to show that the violation of the returned solution is small in $\epsilon$. This is given in Lemma~\ref{lem-trunc2} below, which shows that the solution $\widehat x$ to the rounded problem violates the capacity constraint by only a factor of at most $(1+4\epsilon)$. Both lemmas can be proved in the same way as in \cite{CKM15}; we include the proof below for completeness. 
\end{proof}

For any binary vector $x$ feasible for \textsc{($m$-CSP)}, let us write for brevity  
\begin{gather*}
P_{+,t}(x)\triangleq\sum_{k\in \cN_+}\sum_{j\in D_k:t \in T_{j}} x_{k,j}s_{k,j}^{\rm R}(t), \quad  P_{-,t}(x)\triangleq\sum_{k\in \cN_-}\sum_{j\in D_k:t \in T_{j}}-x_{k,j}s_{k,j}^{\rm R}(t), \\ \text{ and }
P_{I,t}(x)\triangleq\sum_{k\in \cN}\sum_{j\in D_k:t \in T_{j}}x_{k,j}s_{k,j}^{\rm I}(t).
\end{gather*} 
Also, write
\begin{gather*}
\widehat P_{+,t}(x)\triangleq\sum_{k\in \cN_+}\sum_{j\in D_k:t \in T_{j}} x_{k,j} \widehat s_{k,j}^{\rm R}(t), \qquad  \widehat P_{-,t}(x)\triangleq \sum_{k\in \cN_-}\sum_{j\in D_k:t \in T_{j}} -  x_{k,j}\widehat s_{k,j}^{\rm R}(t), \text{ and} \\
\widehat P_{I,t}( x)\triangleq\sum_{k\in \cN}\sum_{j\in D_k:t \in T_{j}}  x_{k,j}\widehat s_{k,j}^{\rm I}(t).
\end{gather*}  
Using the fact that $\ell \le \tau \lceil \frac{\ell}{\tau} \rceil  \le \ell + \tau$ for any $\ell,\tau$ such that $\tau>0$, and that $\sum_{j \in D_k} x_{k,j} \le 1$ by \eqref{mc2}, we have 
\begin{align*}
\widehat P_{+,t}(x) &= \sum_{k\in \cN_+}\sum_{j\in D_k: t \in T_{j} } x_{k,j} \widehat s_{k,j}^{\rm R}(t) \\
&\le \sum_{k\in \cN_+}\sum_{j\in D_k: t \in T_{j} } x_{k,j}  (s_{k,j}^{\rm R}(t)+L_t)= P_{+,t}(x) + nL_t.
\end{align*}
The same bound holds for $\widehat P_{-,t}$ and $\widehat P_{I,t}$:
\small{
\begin{eqnarray}\label{eq:bds}
&\max\{\widehat P_{+,t}(\widehat x)-nL_t,0\}\le P_{+,t}(x)\le \widehat P_{+,t}(\widehat x),~\max\{\widehat P_{-,t}(\widehat x)-nL_t,0\}\le P_{-,t}(x)\le \widehat P_{-,t}(\widehat x), \notag \\ & \max\{\widehat P_{I,t}(\widehat x)-nL_t,0\} \le P_{I,t}(x)\le\widehat P_{I,t}(\widehat x).
\end{eqnarray}}

\begin{lemma}
	\label{lem-trunc}
	For any feasible solution $x$ to \textsc{$m$-CSP $[0,\pi\mbox{-}\varepsilon]$}, we have $$\Big| \sum_{k\in\cN} \sum_{j\in D_k: t \in T_{j} } x_{k,j} \widehat s_{k,j}(t)\Big| \le  ( 1 + 2\epsilon)C_t \qquad \text{for all } t \in \cT.$$
\end{lemma}
\begin{proof}
	Using~\raf{eq:ubounds} and~\raf{eq:bds}, for all $t\in \cT$,
	\begin{align*}
	&\Big( \sum_{k\in\cN} \sum_{j\in D_k: t \in T_{j} } x_{k,j} \widehat s_{k,j}^{\rm R}(t)\Big)^2 +  \Big(\sum_{k\in\cN} \sum_{j\in D_k: t \in T_{j} } x_{k,j} \widehat s_{k,j}^{\rm I}(t) \Big)^2 \\
	&\qquad=\left(\widehat P_{+,t}(x)- \widehat P_{-,t}(x)\right)^2 +   \widehat P_{I,t}^2( x)\nonumber\\
	&\qquad=\widehat P_{+,t}^2(x)+\widehat P_{-,t}^2(x)-2\widehat P_{+,t}(x) \widehat P_{-,t}(x)+ \widehat P_{I,t}^2(x)\nonumber\\ 
	&\qquad\le (P_{+,t}(x)+nL_t)^2+(P_{-,t}(x)+nL_t)^2-2P_{+,t}(x) P_{-,t}(x)+(P_{I,t}(x)+nL_t)^2\nonumber\\
	&\qquad=(P_{+,t}(x)-P_{-,t}(x))^2+P_{I,t}^2(x)+2nL_t (P_{+,t}(x)+P_{-,t}(x)+P_{I,t}(x))+3n^2L_t^2\nonumber\\
	&\qquad= \Big(\sum_{k\in\cN} \sum_{j\in D_k: t \in T_{j} } x_{k,j} s_{k,j}^{\rm R}(t) \Big)^2 + \Big(\sum_{k\in\cN} \sum_{j\in D_k: t \in T_{j} } x_{k,j} s_{k,j}^{\rm I}(t)  \Big)^2 \\
	&\qquad\qquad +2nL_t\Big(\sum_{k\in\cN} \sum_{j\in D_k: t \in T_{j} } x_{k,j}|s_{k,j}^{\rm R}(t)|+\sum_{k\in\cN} \sum_{j\in D_k: t \in T_{j} } x_{k,j} s_{k,j}^{\rm I}(t)\Big)+3n^2L_t^2\nonumber\\
	&\qquad\le C_t^2 + 4nL_t (\tan \theta +1) C_t + 3n^2L_t^2 = C_t^2 + 4 \epsilon C_t^2 + 3\epsilon^2 C_t^2/(1+\tan \theta)^2\nonumber \\ 
	&\qquad\le C_t^2 (1+4\epsilon + 3\epsilon^2)\le C_t^2(1+2\epsilon)^2. 
	\end{align*}
\end{proof}

\medskip

\begin{lemma}
	\label{lem-trunc2}
	Let $\widehat x$ be the solution returned by {\sc $m$-CSP-FPTAS}. Then, $$\Big|\sum_{k\in\cN} \sum_{j\in D_k: t \in T_{j} } \widehat x_{k,j} s_{k,j}(t)\Big|\le(1+4\epsilon) C_t \qquad \text{for all } t \in \cT.$$ 
\end{lemma}
\begin{proof}
	As in the proof of Lemma~\ref{lem-trunc}, for all $t\in \cT$,
	\begin{align}
	\label{eq:ee1}
	&\Big( \sum_{k\in\cN} \sum_{j\in D_k: t \in T_{j} } \widehat x_{k,j} s_{k,j}^{\rm R}(t) \Big)^2 +  \Big(\sum_{k\in\cN} \sum_{j\in D_k: t \in T_{j} } \widehat x_{k,j}s_{k,j}^{\rm I}(t) \Big)^2=\left(P_{+,t}(\widehat x)- P_{-,t}(\widehat x)\right)^2 +  P_{I,t}^2(\widehat x)\nonumber\\
	&=P_{+,t}^2(\widehat x)+P_{-,t}^2(\widehat x)-2P_{+,t}(\widehat x)P_{-,t}(\widehat x)+P_{I,t}^2(\widehat x).
	\end{align} 
	If both $\widehat P_{+,t}(\widehat x)$ and $\widehat P_{-,t}(\widehat x)$  are less than $nL_t$, then the R.H.S. of \raf{eq:ee1} can be bounded by 
	\begin{eqnarray}\label{eq:ee2}
	\widehat P_{+,t}^2(\widehat x)+\widehat P_{-,t}^2(\widehat x)+\widehat P_{I,t}^2(\widehat x)
	&\le&\widehat P_{+,t}^2(\widehat x)+\widehat P_{-,t}^2(\widehat x)-2\widehat P_{+,t}(\widehat x) \widehat P_{-,t}(\widehat x)+2n^2L_t^2+\widehat P_{I,t}^2(\widehat x)\nonumber\\
	&=&(\widehat P_{+,t}(\widehat x)- \widehat P_{-,t}(\widehat x))^2+\widehat P_{I,t}^2(\widehat x)+2n^2L_t^2.
	\end{eqnarray}
	Otherwise, we bound the R.H.S. of  Eqn.~\raf{eq:ee1} by
	\begin{align}\label{eq:ee3}
	&\widehat P_{+,t}^2(\widehat x)+\widehat P_{-,t}^2(\widehat x)-2(\widehat P_{+,t}(\widehat x)-nL_t)(\widehat P_{-,t}(\widehat x)-nL_t)+\widehat P_{I,t}^2(\widehat x)\nonumber\\
	&\qquad=(\widehat P_{+,t}(\widehat x)-\widehat P_{-,t}(\widehat x))^2+\widehat P_{I,t}^2(\widehat x)+2nL_t(\widehat P_{+,t}(\widehat x)+ \widehat P_{-,t}(\widehat x))-2n^2L_t^2.
	\end{align}
	Since $\widehat x=y_+ +  y_-$ is obtained from feasible solutions $y_+$ and $y_-$ to \\$\text{\sc $2m$DKP-Exact}\Big(\{u_{k,j},(\widehat s_{k,j}(t)/L_t)_t\}_{k\in\cN_+, j \in D_k}, \big(\xi_+(t)/L_t\big)_t,\big(\zeta_+(t)/L_t\big)_t\Big)$
	and \\$\text{\sc $2m$DKP-Exact}\Big(\{u_{k,j},(-\widehat s_{k,j}(t)/L_t)_t\}_{k\in \cN_-, j \in D_k}, \big(\xi_-(t)/L_t\big)_t,\big(\zeta_-(t)/L_t\big)_t \Big)$, respectively, and  $\xi_+ , \xi_-,\zeta_+, \zeta_-$ satisfy the condition in Step~\ref{cond1} of Algorithm~\ref{CSP-biFPTAS}, it follows from \raf{eq:ee1}-\raf{eq:ee3} that
	\begin{align*}
	&\Big( \sum_{k\in\cN} \sum_{j\in D_k: t \in T_{j} } \widehat x_{k,j} s_{k,j}^{\rm R}(t) \Big)^2 +  \Big(\sum_{k\in\cN} \sum_{j\in D_k: t \in T_{j} }  \widehat x_{k,j} s_{k,j}^{\rm I}(t) \Big)^2\\
	&\qquad \le \Big( \sum_{k\in\cN} \sum_{j\in D_k: t \in T_{j} }  \widehat x_{k,j} \widehat s_{k,j}^{\rm R}(t)\Big)^2 +  \Big(\sum_{k\in\cN} \sum_{j\in D_k: t \in T_{j} } \widehat x_{k,j} \widehat s_{k,j}^{\rm I}(t) \Big)^2\\
	&\qquad\qquad+2nL_t\sum_{k\in\cN} \sum_{j\in D_k: t \in T_{j} } \widehat x_{k,j} |\widehat s_{k,j}^{\rm R}(t)|+2n^2L_t^2\nonumber\\
	&\qquad=(\xi_+(t) - \xi_-(t))^2 + (\zeta_+(t) + \zeta_-(t))^2 + 2nL_t(\xi_+(t)+\xi_-(t))+2n^2L_t^2\nonumber\\
	&\qquad\le \Big((1+2\epsilon)^2C_t^2 +4n\frac{\epsilon C_t}{n(\tan \theta+1)} (1+\tan \theta)C_t + 2n^2\frac{\epsilon^2C_t^2}{n^2(\tan \theta+1)^2}\Big)\\
	&\qquad\le \Big((1+2\epsilon)^2 +4\epsilon + 2\epsilon^2\Big)C_t^2\le (1+4\epsilon)^2C_t^2.
	\end{align*}
\end{proof}

\section{{$m$-CSP$[0,\tfrac{\pi}{2}]$} with Polynomial number of   Time Slots}\label{sec:poly}

In this section, we extend our results to  polynomial number of time slots $|\cT|$. We assume herein that all demands lie in the first quadrant of the complex plane (i.e., $\phi \triangleq \max_k\arg(s_{k,j}(t))\le \frac{\pi}{2}$ for $\forall~ t \in \cT$).
We provide a reduction to the unsplittable flow problem on a path with bag constraints \mbox{\sc (bag-UFP)} for which recent approximation algorithms are developed in the literature (see, e.g., \cite{Grandoni2015, chakaravarthy2014improved, elbassioni2012approximation}). We remark that {\sc bag-UFP} considers only real-valued demands, whereas in {\sc $m$-CSP} demands are complex-valued.
 We will show that such reduction will increase the approximation ratio of {\sc bag-UFP} by a constant factor of $\cos \tfrac{\phi}{2}$, where $\phi\le \tfrac{\pi}{2}$ is the maximum argument of any demand.
We will  need the following further assumption to accommodate the setting of {\sc bag-UFP}:
\begin{itemize}
	\item all demands are constant over time: $s_{k,j}(t) = s_{k,j}(t')$ for any $t,t' \in T_j$.
	 To simplify notation, let $s_{k,j}$ denote the unique demand over all time steps $T_j$. 
\end{itemize}
For convenience, we shall refer to the problem as {\sc $m$-CSP$'$} when restricted with the above assumption. When all demands in {\sc $m$-CSP$'$} are real-valued, the problem is  called \mbox{\sc bag-UFP}.
We can approximate an instance of {\sc $m$-CSP$'$} by an instance of ($\textsc{bag-UFP}$) defined as follows:
\vspace{-5pt}
\begin{align}
(\textsc{bag-UFP})&\quad \max_{x_{k,j} \in \{0,1\}} \sum_{k \in \cN} \sum_{j \in D_k} u_{k,j} x_{k,j} \nonumber\\
\text{s.t.} \quad&  \sum_{k \in \cN}\sum_{j \in D_k: t  \in T_j} |s_{k,j}| x_{k,j}  \le    C_{t}, \quad \text{ for all } t \in \cT\label{rcon}\\
&\sum_{j \in D_k} x_{k,j} \le 1 \qquad \text{ for all } k \in \cN.
\end{align}\vspace{-10pt}

Note that the absolute of the sum in Cons.~\raf{mc1} is replaced in {\sc bag-UFP} by the sum of the absolutes in Cons.~\raf{rcon}. Thus all demands in {\sc bag-UFP} are real-valued.

	We denote by {$\textsc{$m$-CSP}^{\ast}$} (resp., $\textsc{bag-UFP}^{\ast}$) the linear relaxation of  {$\textsc{$m$-CSP}'$} (resp., {$\textsc{bag-UFP}$}), that is, when $x_{k,j} \in [0,1]$ for all $k\in  \cN$, $j\in D_k$. 
		Let $\OPT$ and $\overline{\OPT}$ be the optimal objective values of {\sc $m$-CSP$'$} and {\sc bag-UFP} respectively. Also denote by $\OPT^*$ and $\overline{\OPT}^*$ the optimal objective value of $\textsc{$m$-CSP}^{\ast}$ and $\textsc{bag-UFP}^{\ast}$, respectively.
		
		We will show in Lemma~\ref{lem:nba} and  Theorem~\ref{thm:poly1} below that one can use the algorithms developed for {\sc bag-UFP} with bounded integrality gap to obtain approximate solutions to {\sc $m$-CSP$'[0,\frac{\pi}{2}]$}.

\begin{lemma}\label{lem:nba}
	Given a solution $\bar x \in \{0,1\}^{|\cI|}$  to  {\sc bag-UFP} such that  $u(\bar x) \ge \psi \cdot \overline\OPT^\ast$, $\psi \in [0,1]$, then $\bar x$ is feasible for {\sc $m$-CSP$'[0,\frac{\pi}{2}]$} and
	$u(\bar x) \ge \psi \cos\tfrac{\phi}{2}   \cdot \OPT.$
\end{lemma}
\begin{proof}
	%
	Let $(x^\ast_{k,j})_{k \in \cN, j \in D_k}$ be an optimal solution for $\textsc{$m$-CSP}^{\ast}$.  
	Lemma~\ref{lem:tb} below implies that
	\begin{equation*}
	\cos\tfrac{\phi}{2} \cdot \sum_{k \in \cN} \sum_{j \in D_k: t \in T_j} |s_{k,j}|  x_{k,j}^\ast \le \Big| \sum_{k \in \cN} \sum_{j \in D_k: t \in T_j} s_{k,j}  x_{k,j}^\ast\Big| \le   C_t \qquad \forall t \in \cT.
	\end{equation*}	
	According to the above inequality, we can construct a feasible  solution $(\tilde x_{k,j})_{k \in \cN, j \in D_k}$ to $\textsc{bag-UFP}^{\ast}$ defined by \mbox{$\tilde x_{k,j} \triangleq \cos\tfrac{\phi}{2}\cdot  x_{k,j}^\ast$}.
	By the feasibility of $(\tilde x_{k,j})_{k \in \cN, j \in D_k}$,
	
	\begin{equation*}
	\overline\OPT^\ast \ge \sum_{k \in \cN} \sum_{j \in D_k} u_{k,j} \tilde x_{k,j} = \cos\tfrac{\phi}{2} \cdot \sum_{k \in \cN}\sum_{j \in D_k} u_{k,j} x^\ast_{k,j} = \cos\tfrac{\phi}{2} \cdot \OPT^\ast. \label{eq:rst}
	\end{equation*}
	Therefore, $u(\bar x) \ge \psi \cdot \overline\OPT^\ast\ge \psi \cdot \cos\tfrac{\phi}{2} \cdot \OPT^\ast \ge  \psi  \cos\tfrac{\phi}{2} \cdot \OPT$.
	
	It remains to show that $\bar x$ is feasible for {\sc $m$-CSP$'$}, which follows readily from the triangular inequality: 
	\begin{equation*}
	\Big| \sum_{k \in \cN} \sum_{j \in D_k : t \in T_j} s_{k,j}  \bar x_{k,j}\Big| \le   \sum_{k \in D_k} \sum_{j \in D_k : t \in T_j} |s_{k,j}|  \bar x_{k,j} \le   C_t \qquad \forall t \in \cT.
	\end{equation*}	
\end{proof}
\begin{lemma}[\cite{KKCMZ16}]
	\label{lem:tb}
	Given a set of vectors $\{d_i \in \RR^2\}_{i=1}^n$, then 
	$ \frac{\sum_{i=1}^n |d_i| }{\big| \sum_{i =1}^n d_i \big|} \le \sec\tfrac{\theta}{2},$ 
	where $\theta$ is the maximum angle between any pair of vectors  $\{d_i \in \RR^2\}_{i=1}^n$ and $0 \le \theta \le \frac{\pi}{2}$.
\end{lemma}
For completeness, we provide the proof in the appendix.

We can apply Lemma~\ref{lem:nba} using the recent LP-based algorithm by Grandoni et al.~\cite{Grandoni2015} to obtain the following result.
\begin{theorem}\label{thm:poly1}
	There exists an   $(\Omega(\log\log n/\log n),1)$-approximation for $\textsc{$m$-CSP}'[0,\frac{\pi}{2}]$. Additionally, if all demands have the same utility, we obtain $(\Omega(1),1)$-approximation.
\end{theorem}

Prior work has addressed an important restriction of  {\sc UFP} (also {\sc bag-UFP}) called the no bottleneck assumption (NBA), namely,  $\max_{k \in \cN,~j\in D_k} |s_{k,j}| \le C_{\min} \triangleq \min_{t \in \cT} C_t$, that is, the largest  demand is at most the smallest capacity over all time slots. Define the bottleneck time of demand $(k,j)$ by $b_{k,j}\triangleq \arg\min_{t\in T_j} C_t$. 
Given a constant $\delta \in [0,1]$, we call a demand $(k,j)$ $\delta$-small if $|s_{k,j}| \le \delta C_{b_{k,j}}$, otherwise we call it $\delta$-large.
We remark that the NBA assumption naturally holds in smart grids since individual demands are typically much smaller than the generation capacity over all time slots.
In the following, we show that there exists an $(\Omega(1),1)$-approximation for {\sc $m$-CSP$'[0,\frac{\pi}{2}]$}, under NBA. This is achieved by splitting demands to  $\delta$-small and $\delta$-large and solving each instance separately then taking the maximum utility solution. 
The next lemma is an extension to an earlier work by Chakrabarti et al.~\cite{chakrabarti2007approximation} (to accommodate complex-valued demands) used to derive a dynamic program that approximates $\delta$-large demands.
\begin{lemma}\label{lem:delta-large}
	The number of $\delta$-large demands that cross any time slot in any feasible solution is at most $2  \lfloor \frac{1}{\delta^2} \cdot {\sec\tfrac{\phi}{2}}  \rfloor$.
\end{lemma}
\begin{proof}
	Given a feasible solution $\hat x$, let $S\triangleq \{(k,j) \in \cI \mid \hat x_{k,j} = 1, s_{k,j} > \delta b_{k,j}\}$  be the set of indices of $\delta$-large demands. Consider any time slot $t$, let $S_t\triangleq \{(k,j) \in S \mid t \in T_j \}$ be the set of demands that cross time  $t$. Then we partition $S_t$ to the sets $S_t^{\rm L}$ and $S_t^{\rm R}$, such that $S_t^{\rm L}$  (resp., $S_t^{\rm R}$) contains demands with bottleneck time slot on the left (resp., right)  of $t$. We show that $|S_t^{\rm L}| \le \lfloor \frac{1}{\delta^2} \cdot {\sec\tfrac{\phi}{2}}  \rfloor$, and a similar argument shows  the same bound for $|S_t^{\rm R}|$.
	
	Let $B$ be the set of bottleneck time slots for demands in $S_t^{\rm L}$. Now let $t' \in B$ be the rightmost bottleneck time slot in $B$. Since $t'$ is the bottleneck of some $\delta$-large demand $(k,j)$, i.e.,  $\delta C_{t'} < |s_{k,j}|$, and by the NBA assumption, $|s_{k,j}|\le C_{\min}$; it follows that $C_{t'} < \tfrac{C_{\min}}{\delta}$. 
	Because $t'$ is the rightmost time slot in $B$, all demands in $S_t^{\rm L}$ pass through $t'$, therefore $|\sum_{(k,j) \in S_t^{\rm L}} s_{k,j}|\le C_{t'}$. Since all demands $(k,j) \in S_t^{\rm L}$ are $\delta$-large,   $|s_{k,j}| > \delta C_{b_{k,j}} \ge \delta C_{\min}$. Therefore, using Lemma~\ref{lem:tb}
	\begin{align*}
	\delta C_{\min} |S_t^{\rm L}| <   \sum_{(k,j) \in S_t^{\rm L}} |s_{k,j}| \le   \sec \tfrac{\phi}{2} \Big| \sum_{(k,j) \in S_t^{\rm L}} s_{k,j}\Big| \le \sec \tfrac{\phi}{2} \cdot C_{t'} <\sec \tfrac{\phi}{2} \cdot  \tfrac{1}{\delta}C_{\min}.
	\end{align*}
	This gives $|S_t^{\rm L}| \le \lfloor \frac{1}{\delta^2} \cdot {\sec\tfrac{\phi}{2}}  \rfloor$. 
\end{proof}

\begin{theorem}\label{thm:poly2}
	Under the NBA assumption, there exists an  $(\Omega(1),1)$-approximation for $\textsc{$m$-CSP}'[0,\frac{\pi}{2}]$. The running time is $O(n^2)$.
\end{theorem}
\begin{proof}
	We set $\delta = \tfrac{1}{2}$. For small demands, Chakaravarthy et al. \cite{chakaravarthy2014improved} present a primal-dual $\tfrac{1}{9}$-approximation algorithm for {\sc bag-UFP} that runs in $O(n^2)$. 
	By Lemma~\ref{lem:nba}, this algorithm is  also $(\tfrac{1}{9}\cos \tfrac{\phi}{2})$-approximation to {\sc $m$-CSP$'$} with small demands. Large demands can be handled via a simple reduction to the weighted job interval selection problem proposed in~\cite{elbassioni2012approximation} (i.e., restrict selected demands to be disjoint). By Lemma~\ref{lem:delta-large} we loose a factor of $8 \sec \tfrac{\phi}{2}$, then we apply the $2$-approximation by Bar-Noy et al.~\cite{bar2001unified} which runs in $O(n \log n)$. Hence we obtain a $(\frac{1}{25}\cos \tfrac{\phi}{2},1)$-approximation that runs in $O(n^2)$.
\end{proof}

\section{Practical Greedy Approximation for {$1$-CSP}$[0,\frac{\pi}{2}]$} \label{sec:greedy}
%
%

In this section we give a practical greedy constant-factor approximation algorithm, presented in Algorithm~\ref{alg:gpa}, and denoted by {\sc 1-CSP-Greedy}, for the single time slot case ({\sc $1$-CSP}$[0,\frac{\pi}{2}]$) where $|\cT| = 1$. Despite the theoretical value of the PTAS and FPTAS presented in \cite{CKM14} (that are generalized in Sec.~\ref{sec:const}), the running time is quite large and hence impractical for real world applications. Algorithm {\sc 1-CSP-Greedy}, on the other hand, achieves $\big(\frac{1}{2}\cos \frac{\phi}{2},1\big)$-approximation in ${O}(N \log N)$ time, where $N \triangleq \sum_{k\in \cN} |D_k|$. 
This result can be derived directly by combining Lemma~\ref{lem:nba} (restricted to the case $|\cT|=1$) with the known analysis of the greedy algorithm for the multiple-choice Knapsack problem \cite{IH78} and its connection to the LP relaxation. However, we include the proof here for completeness.
Note that such a simple greedy algorithm can be 
used to provide a fast heuristic for practical settings when considering multiple time slots.
For instance, in the setting where users arrive online, {\sc 1-CSP-Greedy} could be applied to each time slot, after reducing the capacity by the magnitude of demands consumed in previous time slots.

Consider the simplified version of  {\sc $m$-CSP} denoted by {\sc 1-CSP} where $|\cT|=1$: 
	
	\begin{align}
	\textsc{(1-CSP)} \quad& 
	\displaystyle \max  \sum_{k\in \cN } \sum_{j \in D_k} u_{k,j} x_{k,j} \\
	\text{subject to }  \quad  &  \displaystyle\Big|\sum_{k\in \cN} \sum_{j\in D_k} s_{k,j}\cdot x_{k,j}\Big| \le C \label{stmc1}\\
	&\sum_{j \in D_k} x_{k,j} \le 1, \qquad \text{ for all } k \in \cN \label{stmc2}\\
	& x_{k,j}\in\{0,1\} \text{ for all }(k,j)\in \cI. \label{stmc3}
	\end{align} 

For convenience, we add a dummy demand to each set $D_k$, for each user $k\in \cN$ with utility of $0$ and demand of $s_{k,0}=\bzero$. This is to guarantee that a solution to {\sc 1-CSP} problem contains exactly one demand from each set $D_k$ for every user $k \in \cN$. Note that this change does not affect the {\sc 1-CSP} problem.
 
 If a user's complex-valued power demand is substituted in ({\sc 1-CSP}) by its real-valued magnitude, the inequality constraint~\raf{stmc2} is transformed into an equality constraint and the binary decision variables $x_{k,j}$ are relaxed such that they take non-negative real values (i.e., $(x_{k,j})_{k \in \cN, j \in D_k} \in [0, 1]^{|\cI|}$), the following linear programming (LP) problem is obtained.
 
 \begin{align}
 \textsc{(Rx1-CSP)} \quad& 
 \displaystyle \max  \sum_{k\in \cN } \sum_{j \in D_k} u_{k,j} x_{k,j} \\
 \text{subject to }  \quad  &  \displaystyle \sum_{k\in \cN} \sum_{j\in D_k} \Big| s_{k,j} \Big| \cdot x_{k,j} \le C \label{rxstmc1}\\
 &\sum_{j \in D_k} x_{k,j} = 1, \qquad \text{ for all } k \in \cN \label{sxsmc2}\\
 & x_{k,j}\in[0,1] \text{ for all }(k,j)\in \cI \, . \label{rxstmc3}
 \end{align}
 



We make use of the following statement.
\begin{proposition}[\cite{CHW75,IH78,KPP10book}]\label{prp:extme} 
	\begin{itemize}
	\item[(i)] If two demands $j, h \in D_k$ belonging to the same set $D_k$, for $k \in \cN$, with $|s_{k, j}| \leq |s_{k, h}|$ satisfy \begin{equation*}\label{def:dom}
	u_{k, j} \geq u_{k, h} \, ,
	\end{equation*}
	then an optimal solution to ({\sc 1-CSP}) with $x_{k, h} = 0$ exists.
   \item[(ii)] If two demands $j, h \in D_k$ belonging to the same set $D_k$, for $k \in \cN$, with $|s_{k, j}| \leq |s_{k, h}|$ satisfy
   \begin{equation*}\label{def:dom}
   \frac{u_{k, j}}{|s_{k, j}|} \leq \frac{u_{k, h}}{|s_{k, h}|} \, ,
   \end{equation*}
   then an optimal solution to {\sc Rx1-CSP} with $x_{k, j} = 0$ exists.
   \item[(iii)] If some demands $j, h, l \in D_k$, $k \in \cN$ with $|s_{k, j}| < |s_{k, h}| < |s_{k, l}|$, $u_{k, j} < u_{k, h} < u_{k, l}$, and $\frac{u_{k, j}}{|s_{k, j}|} \geq \frac{u_{k, h}}{|s_{k, h}|}\ge\frac{u_{k, l}}{|s_{k, l}|}$ satisfy
   \begin{equation*}\label{def:dom2}
     \frac{u_{k, h} - u_{k, j}}{|s_{k, h}| - |s_{k, j}|}\leq\frac{u_{k, l} - u_{k, h}}{|s_{k, l}| - |s_{k, h}|} \, ,
   \end{equation*}
   then an optimal solution to {\sc Rx1-CSP} with $x_{k, h} = 0$ exists.
\end{itemize}
\end{proposition}

The above proposition implies that, without losing all optimal solutions to {\sc Rx1-CSP}, we can preprocess the demands of each set $D_k$, $k \in \cN$, to obtain a corresponding new set $R_k \subseteq D_k$ that satisfies:
\begin{align}\label{e1}
&|s_{k, 1}| \leq |s_{k, 2}| \leq ...\leq |s_{k, r_k}|,\\
&|u_{k, 1}| \leq |u_{k, 2}| \leq ...\leq |u_{k, r_k}|, \label{e2}\\
&\frac{u_{k, 1}}{|s_{k, 1}|} > \frac{u_{k, 2}}{|s_{k, 2}|} > ...> \frac{u_{k, r_k}}{|s_{k, r_k}|},\label{e3}\\
&\frac{u_{k, 2} - u_{k, 1}}{|s_{k, 2}| - |s_{k, 1}|} > \frac{u_{k, 3} - u_{k, 2}}{|s_{k, 3}| - |s_{k, 2}|} > ...> \frac{u_{k, r_k} - u_{k, r_k-1}}{|s_{k, r_k}| - |s_{k, r_{k-1}}|}.\label{e4}
\end{align}
Observe that this reduction requires only ${O}(\sum_{k \in \cN}|D_k|\log |D_k|)$ time, as it can be done by sorting each $D_k$ followed by a linear scan to remove the demands that do not appear in the optimal solution.

In \cite{IH78} (see also \cite[Chapter 11]{KPP10book}), it was also proved that the LP optimal solution to {\sc Rx1-CSP} problem may be found by a greedy algorithm which starts by finding the sets $R_k$ above. Assume the ordering $|s_{k, 1}| \leq |s_{k, 2}| \leq ...\leq |s_{k, r_k}|$ in $R_k$, where $r_k = |R_k|$.
Initially, the algorithm selects the dummy demand $s_{k,0}$ for each customer and sets the corresponding decision variables to $1$.   
Next, the greedy algorithm constructs a new set $E$ by combining all the sets $R_k$, $k \in \cN$ and setting $\tilde{u}_{k, j} = u_{k, j} - u_{k, j-1}$ and $|\tilde{s}_{k, j}| = |s_{k, j}| - |s_{k, j-1}|$ for $j = 1, ..., r_k$. After sorting entries in $E$ by their {\it efficiency}, defined as $\frac{\tilde{u}_{k, j}}{|\tilde{s}_{k, j}|}$ in non-increasing order, the greedy execution continues by selecting demands in the aforementioned sorted considering the capacity $C$. Each time an item $(k, j)$ is selected from set $E$, we assign $\tilde{x}_{k, j} \leftarrow 1$, $\tilde{x}_{k, j-1} \leftarrow 0$ and $\tau = \tau + |\tilde{s}_{k, j}|$, where the initial value of $\tau$ is $0$.  
	Assume at some iteration adding the next item $(k', j')$ to the current solution vector $\tilde{x}$ causes capacity violation, that is
	
	\begin{equation}
	\displaystyle \tau \le C \text{  and }  \tau + |\tilde{s}_{k', j'}| > C \, .
	\end{equation}
	
	The greedy execution is stopped at this point and the remaining capacity $C - \tau $ is occupied by the corresponding fractional part of the $(k', j')$ item's power demand and the item's $(k', j'-1)$ decision variable is set as follows: $$ \tilde{x}_{k', j'} = \frac{C  - \tau}{|\tilde{s}_{k', j'}|} \text{ and } \tilde{x}_{k', j'-1} = 1 - \tilde{x}_{k', j'} \, \text{ where } k' \in \cN, j' \in R_{k'}.$$
	
	In \cite{CHW75,IH78}, it was shown that this greedy strategy indeed produces an optimal solution to {\sc Rx1-CSP} problem containing at most two fractional variables that belong to adjacent users in the sorted set $R_{k'}$ as given above.  
 Note that algorithm {\sc 1-CSP-Greedy} is almost the same as this greedy algorithm algorithm described above, except that we drop the fractional values.
 
\begin{algorithm}[htb!]
	\caption{{\sc 1-CSP-Greedy}$ [\{u_{k,j}, s_{k,j}\}_{k\in\cN, j\in D_k}, C]$}
	\begin{algorithmic}[1]
		\Require Users' utilities and demands $\{u_{k,j}, s_{k,j}\}_{k\in\cN, j\in D_k}$; capacity $C$
		\Ensure $(\frac{1}{2}\cos \frac{\phi}{2},1)$-solution $\bar{x}$ to \textsc{1-CSP}
		\Initialize
		\Statex{$\bullet$} Add a dummy demand with zero utility  and zero demand to each set $D_k$, $k \in \cN$
		\Statex{$\bullet$} Sort users in each set $D_k$, $k \in \cN$ by the magnitude of their demands in increasing order such that if $j \le j'$, then $|s_{k, j}| \le |s_{k', j'}| \text{ for all } j', j \in D_k$
		\Statex{$\bullet$} For each $k \in \cN$ define a new set $R_k \subseteq D_k$ by successively testing the demands in  $D_k$ according to Eqns.~\raf{e1}-\raf{e4}. Assume the ordering $|s_{k, 1}| \leq |s_{k, 2}| \leq ... \leq |s_{k, r_k}|$ in $R_k$, where $r_k = |R_k|$
		\Statex{$\bullet$} $E \leftarrow \varnothing$, $\tilde{x}\leftarrow \bzero$, $\tilde{x}_{k, 0}\leftarrow 1$ for all $k \in \cN$, $\tau \leftarrow0$, $\widehat{x}\leftarrow \tilde{x}$
		\For{$k \in \cN$, $j = 1, ..., r_k$}

		\State $\tilde{u}_{k, j} \leftarrow u_{k, j} - u_{k, j-1}$, $\tilde{s}_{k, j} \leftarrow s_{k, j} - s_{k, j-1}$
		\State $E \leftarrow E \cup \{(k, j)\}$
		\EndFor
		\State Sort items in $E$ by their efficiency ($\frac{\tilde{u}_{k, j}}{|\tilde{s}_{k, j}|}$) in a non-increasing order
		\For{$(k, j) \in E$ (in the sorted order)} 
		\If{ $\tau+ \big|\tilde{s}_{k, j} \big| \le C$}
		\State $\tilde{x}_{k, j} \leftarrow 1$, $\tilde{x}_{k, j-1} \leftarrow 0$, $\tau \leftarrow \tau + \big|\tilde{s}_{k, j}\big|$
		\State \textbf{break}
		\EndIf
		\EndFor
		\State Set $ \widehat{x}_{k', j'} \leftarrow 1 \text{ for } (k', j') \triangleq \arg\max_{j \in R_{k}, k \in \cN} \{u_{k, j}\}$, $ \widehat{x}_{k', 0} \leftarrow 0$
		\State Set $\bar{x} \leftarrow \arg\max_{ x \in \{\widehat{x}, \tilde{x}\}}u(x)$ \label{alg:gh.max}
		\State \Return $\bar{x}$
	\end{algorithmic}
	\label{alg:gpa} 
\end{algorithm}


\begin{theorem} 
	Algorithm {\sc 1-CSP-Greedy} is $\big(\frac{1}{2} \cos \frac{\phi}{2},1\big)$-approximation for \textsc{1-CSP$[0,\frac{\pi}{2}]$}. The running time is $O(N\log N)$.
	\label{th:gpa}\vspace{-5pt}
\end{theorem}

\begin{proof}

	Let $ S^\ast\subseteq {\cal I}$ be an optimal solution of ({\sc 1-CSP}), and denote by $\OPT$ and $\OPT^*$ the optimal objective values of ({\sc 1-CSP}) and ({\sc Rx1-CSP}), respectively.
	Denote by $E_s \triangleq E \, \setminus \{(k', j'), (k', j'-1)\}$, and let 
	\begin{equation}
	\widehat{p} \triangleq \sum\limits_{(k, j) \in E_s} u_{k, j} \tilde{x}_{k, j} \text{  and  } u_{\max} \triangleq \max_{j \in R_{k}, k \in \cN} \{u_{k, j}\} \, ,
	\end{equation}
	where $\tilde{x}$ is as defined in the algorithm. 
	For the optimal solution to {\sc Rx1-CSP} problem we get
	\begin{equation}
	\label{eq:sxoptl}
	\displaystyle \OPT^* =  \hat{p} +  \tilde{x}_{k', j'} u_{k', j'} + \tilde{x}_{k', j'-1} u_{k', j'-1} \leq \hat{p}  + u_{k', j'} \, .
	\end{equation}
	
	On the other hand, by Lemma \ref{lem:tb} it follows that
	
	\begin{equation}
	\displaystyle \cos \tfrac{\phi}{2}\cdot \sum_{(k, j) \in S^\ast} |s_{k, j}| \le \Big| \sum_{(k, j) \in S^\ast} s_{k, j} \Big | \le C \, .
	\end{equation}
	
Note that the subset $S^{\ast}$, which is an optimal solution to ({\sc 1-CSP}), becomes a feasible solution to {\sc Rx1-CSP} if  the relaxed decision variables are set $x_{k, j}=\cos \tfrac{\phi}{2}$ for all $(k, j) \in S^{\ast}$, $x_{k,0}=1-\sum_{(k,j)\in S^*}x_{k,j}$ and $x_{k, j}=0$ otherwise. 
This implies that
	
	\begin{equation}
	\label{eq:rxoptl4}
	\displaystyle \OPT^* \ge \cos \tfrac{\phi}{2} \cdot \sum_{(k, j) \in S^\ast} u_{k, j} = \cos \tfrac{\phi}{2} \cdot \OPT \,.
	\end{equation}
	
	Denote by $Z^{\textsc{Alg}}$ the utility of the output solution of {\sc 1-CSP-Greedy} when applied to {\sc 1-CSP} problem. To investigate the worst case approximation ratio of {\sc 1-CSP-Greedy} for {\sc 1-CSP} problem, consider Eqn.~\raf{eq:sxoptl} and observe that
	\begin{equation}
	\label{eq:rxoptlc2}
	\displaystyle \OPT^*\le \hat{p} + u_{\max} \, .
	\end{equation}
	
	Evidently, $Z^{\textsc {Alg}} \ge \hat{p}$. This gives
	\begin{equation}
	\label{eq:rxoptl2}
	\displaystyle \OPT^* \le Z^{\textsc {Alg}} + u_{\max} \,.
	\end{equation}
	
	From the formulation of algorithm {\sc 1-CSP-Greedy}, $Z^{\textsc {Alg}} \ge u_{\max}$, and hence by Eqns.~\raf{eq:rxoptl2} and \raf{eq:rxoptl4} it follows that
	\begin{equation}
	\label{eq:rxoptl5}
	\displaystyle Z^{\textsc Alg}  \ge \frac{1}{2} \cos \tfrac{\phi}{2} \cdot \OPT \,.
	\end{equation}
	
	Finally, note that the solution $\bar x$ is feasible for ({\sc 1-CSP}) by the triangular inequality. 
%
\end{proof}
\section{Extension to the Mixed Case}\label{sec:reduction}
In practical applications of the complex-demand scheduling  problem, we may have the situation when some of the users' demands are elastic in the sense that they can be partially satisfied. An example is an appliance that should be either supplied with a fixed amount of power, or switched off. Formally, we may assume that each user's demand is composed of two sets $D_k=D_k^{\cI}\cup D_k^{\cF}$, where each is a set of demands of the form $\{s_{k,j}(t)\}_{t\in T_j}$, as before. A feasible solution now would select, for each user $k$, one of the demands in $j\in D_k$ and assign either $x_{k,j}\in\{0,1\}$ if $j\in D_k^{\cI}$ or $x_{k,j}\in[0,1]$ if $j\in D_k^{\cF}$.  

We show in this section that we can reduce this mixed case to the fully inelastic case. First, we note that 
\begin{equation}\label{eq:LB}
\OPT\ge LB\triangleq\max\left\{\max_{k\in\cN,~j\in D_k^{\cI}}u_{k,j},\max_{k\in\cN,~j\in D_k^{\cF}}\min\left\{\min_{t\in T_j}\frac{u_{k,j}C_t}{|s_{k,j}(t)|},u_{k,j}\right\}\right\}.
\end{equation}

Let $\cO=\{u_{k,j}, \{s_{k,j}(t)\}_{t\in T_{j}}\}_{k\in\cN, j\in D_k=D_k^{\cI}\cup D_k^{\cF}}$, we construct a fully inelastic instance $\cO'=\{u'_{k,j},$ $\{s'_{k,j}(t)\}_{t\in T_{j}}\}_{k\in\cN, j\in D_k'}$ as follows. Let $\epsilon\in(0,1)$ be an arbitrary constant. For each $k\in\cN$,
we define the set $D_k'=D_k^{\cI}\cup D_k''$, where $D_k''$ is defined as follows. For each $k\in\cN$, $j\in D_k^{\cF}$, we introduce a number of $n_{k,j}=\left\lceil\log_{1+\epsilon}\frac{n u_{k,j}}{\epsilon\cdot LB}\right\rceil$ new demands, given by 
$s_{k,j}^i(t)=\frac{\epsilon\cdot LB(1+\epsilon)^i s_{k,j}(t)}{nu_{k,j}}$ for $t\in T_{j}$, with utility $u_{k,j}^i=\frac{\epsilon\cdot LB(1+\epsilon)^i}{n}$. Then we set $D_k''=\left\{(j,i)~:i=1,\ldots,n_{k,j}~\right\}$, where $(j,i)\in D_k''$ indices the demand $\{s^i_{k,j}(t)\}_{t\in T_{j}}$; we denote the corresponding variable in the formulation \textsc{($m$-CSP)} of the new instance by $x_{k,j}^i$. 

Given a solution $x$ for $\cO'$, we construct a solution $\widehat x$ for $\cO$ in the obvious way: if $j\in D_k^{\cI}$, then we set $\widehat x_{k,j}=x_{k,j}$; otherwise, if $x_{k,j}^i=1$, we set $\widehat x_{k,j}=\frac{\epsilon\cdot LB(1+\epsilon)^i}{nu_{k,j}}$.  
\begin{lemma}
	\label{lem:rel}
Let $x$ be an $(\alpha,\beta)$-approximate solution for $\cO'$. Then $\widehat x$ is a $((1-\epsilon)\alpha,\beta)$-approximate solution for $\cO$. 
\end{lemma}
\begin{proof}
	Let $x^*$ be an optimal solution for $\cO$. We round $x^*$ to a $(1-\epsilon,1)$-approximate solution $\widetilde x$ for $\cO'$ as follows. If $j\in D_k^{\cI}$, we keep $\widetilde x_{k,j}=x^*_{k,j}\in\{0,1\}$. Otherwise, $x^*_{k,j}\in[0,1]$ is positive only for at most one index $j\in D_k^{\cF}$. In this case, we set $\widetilde x_{k,j}=0$ if $x^*_{k,j}<\frac{\epsilon\cdot LB}{n u_{k,j}}$, and otherwise set $\widetilde x_{k,j}=\frac{\epsilon\cdot LB(1+\epsilon)^i}{n u_{k,j}}$, where $i$ is the largest integer $i'$ such that  
	$\frac{\epsilon\cdot LB(1+\epsilon)^{i'}}{n u_{k,j}}\le x^*_{k,j}$. 
	Note that $\widetilde x$ is feasible for $\cO'$ since $\widetilde x\le x^*$. Furthermore, $u(\widetilde x)\ge (1-\epsilon)u(x^*)$, since the total utility corresponding to all the variables that are dropped to $0$ is at most $$\sum_{k\in\cN,~j\in D_k^{\cF}}u_{k,j}\cdot \frac{\epsilon\cdot LB}{n u_{k,j}}\le \epsilon\cdot\OPT=\epsilon\cdot u(x^*),$$ while for all other variables we have $\widetilde x_{k,j}\ge (1-\epsilon)x^*_{k,j}$. Moreover,
	\begin{align*}
	\sum_{k\in \cN} \sum_{j\in D_k: T_{j} \ni t} s_{k,j}(t)\cdot \widehat x_{k,j}&=\sum_{k\in \cN} \sum_{j\in D_k^{\cI}: T_{j} \ni t,~x_{k,j}=1} s_{k,j}(t)\\ &\qquad+\sum_{k\in \cN} \sum_{j\in D_k^{\cF}: T_{j} \ni t,~x_{k,j}^i=1}s_{k,j}(t)\frac{\epsilon\cdot LB(1+\epsilon)^i}{nu_{k,j}}\\	
	&=\sum_{k\in \cN} \sum_{j\in D_k^{\cI}: T_{j} \ni t} s_{k,j}(t)x_{k,j}+\sum_{k\in \cN} \sum_{(j,i)\in D_k'': T_{j} \ni t} s^i_{k,j}(t)x_{k,j}^i\\
	&\le \beta\cdot C_t,
	\end{align*}
	by the $\beta$-feasibility of $x$ for $\cO'$.
	The lemma follows.
\end{proof}

\section{Conclusion}\label{sec:concl}
This paper extends the previous results known for the single time slot case ({\sc CKP}) to a more general scheduling setting.  When the number of time slots $m$ is  constant, both the previously known PTAS and FPTAS are extended to handle multiple-time slots, multiple user preferences, and handle mixed elastic and inelastic demands. For polynomial $m$, a reduction is presented from {\sc CSP$[0,\tfrac{\pi}{2}]$} to the real-valued {\sc bag-UFP}, which can be used to obtain algorithms for {\sc CSP$[0,\tfrac{\pi}{2}]$} based on {\sc bag-UFP} algorithms that have bounded integrability gap for their LP-relaxation. We further presented a practical greedy algorithm that can be  implemented efficiently in real systems.
As a future work, it would be interesting to improve the second case (polynomial $m$) to a constant-factor approximation, following the recent results in \cite{AGLW13}. Additionally, it might be of interest to consider different objective functions such as minimizing the maximum peak consumption at any time slot. Complementing this paper, extended algorithms have been developed for more sophisticated settings, such as online algorithm for {\sc CSP} \cite{KKCE16b} and scheduling in electrical power networks \cite{MCK17, MCK18a, MCK18b, MCK18c}

%

\section*{Acknowledgments}
We thank the anonymous reviewers for careful reading and helpful comments.
\section*{References}

\bibliographystyle{splncs03}
\bibliography{reference}

\begin{thebibliography}{10}
\expandafter\ifx\csname url\endcsname\relax
  \def\url#1{\texttt{#1}}\fi
\expandafter\ifx\csname urlprefix\endcsname\relax\def\urlprefix{URL }\fi
\expandafter\ifx\csname href\endcsname\relax
  \def\href#1#2{#2} \def\path#1{#1}\fi

\bibitem{GS94power}
J.~Grainger, W.~Stevenson, Power System Analysis, McGraw-Hill, 1994.

\bibitem{DR09}
C.-L. Su, D.~Kirschen, Quantifying the effect of demand response on electricity
  markets, Power Systems, IEEE Transactions on 24~(3) (2009) 1199--1207.
\newblock \href {http://dx.doi.org/10.1109/TPWRS.2009.2023259}
  {\path{doi:10.1109/TPWRS.2009.2023259}}.

\bibitem{YC13CKP}
L.~Yu, C.-K. Chau, {Complex-demand Knapsack Problems and Incentives in AC Power
  Systems}, in: Proceedings of the 2013 International Conference on Autonomous
  Agents and Multi-agent Systems, AAMAS '13, Richland, SC, 2013, pp. 973--980.

\bibitem{woeginger2000does}
G.~J. Woeginger, When does a dynamic programming formulation guarantee the
  existence of a fully polynomial time approximation scheme ({FPTAS})?, INFORMS
  Journal on Computing 12~(1) (2000) 57--74.

\bibitem{CKM14}
C.-K. Chau, K.~Elbassioni, M.~Khonji, Truthful mechanisms for combinatorial ac
  electric power allocation, in: Proceedings of the 2014 International
  Conference on Autonomous Agents and Multi-agent Systems, AAMAS '14, Richland,
  SC, 2014, pp. 1005--1012, {http://arxiv.org/abs/1403.3907}.

\bibitem{CKM15}
C.-K. Chau, K.~Elbassioni, M.~Khonji, Truthful mechanisms for combinatorial
  allocation of electric power in alternating current electric systems for
  smart grid, ACM Transactions on Economics and Computation 5 (2016) 7:1--7:29,
  http://arxiv.org/abs/1507.01762.

\bibitem{KCE14}
M.~Khonji, C.~K. Chau, K.~Elbassioni, Inapproximability of power allocation
  with inelastic demands in ac electric systems and networks, in: 2014 23rd
  International Conference on Computer Communication and Networks (ICCCN),
  2014, pp. 1--6.

\bibitem{KT15}
K.~Elbassioni, T.~T. Nguyen, Approximation schemes for multi-objective
  optimization with quadratic constraints of fixed cp-rank, in: Proceedings of
  the 4th International Conference on Algorithmic Decision Theory - Volume
  9346, ADT 2015, Springer-Verlag, Berlin, Heidelberg, 2015, pp. 273--287.

\bibitem{KKCMZ16}
A.~Karapetyan, M.~Khonji, C.~K. Chau, K.~Elbassioni, H.~H. Zeineldin, Efficient
  algorithm for scalable event-based demand response management in microgrids,
  IEEE Transactions on Smart Grid 9~(4) (2018) 2714--2725.
\newblock \href {http://dx.doi.org/10.1109/TSG.2016.2616945}
  {\path{doi:10.1109/TSG.2016.2616945}}.

\bibitem{MCK16}
M.~Khonji, C.~K. Chau, K.~Elbassioni, Optimal power flow with inelastic demands
  for demand response in radial distribution networks, IEEE Transactions on
  Control of Network Systems 5~(1) (2018) 513--524.
\newblock \href {http://dx.doi.org/10.1109/TCNS.2016.2622362}
  {\path{doi:10.1109/TCNS.2016.2622362}}.

\bibitem{darmann2010resource}
A.~Darmann, U.~Pferschy, J.~Schauer, Resource allocation with time intervals,
  Theoretical Computer Science 411~(49) (2010) 4217--4234.

\bibitem{BNC07}
N.~Bansal, A.~Chakrabarti, A.~Epstein, B.~Schieber, A quasi-ptas for
  unsplittable flow on line graphs, in: STOC, ACM, 2006, pp. 721--729.
\newblock \href {http://dx.doi.org/10.1145/1132516.1132617}
  {\path{doi:10.1145/1132516.1132617}}.

\bibitem{anagnostopoulos2014mazing}
A.~Anagnostopoulos, F.~Grandoni, S.~Leonardi, A.~Wiese, A mazing 2+
  $\varepsilon$ approximation for unsplittable flow on a path, in: SODA, SIAM,
  2014, pp. 26--41.

\bibitem{chekuri2007multicommodity}
C.~Chekuri, M.~Mydlarz, F.~B. Shepherd, Multicommodity demand flow in a tree
  and packing integer programs, ACM Transactions on Algorithms (TALG) 3~(3).
\newblock \href {http://dx.doi.org/10.1145/1273340.1273343}
  {\path{doi:10.1145/1273340.1273343}}.

\bibitem{chakaravarthy2010varying}
V.~T. Chakaravarthy, V.~Pandit, Y.~Sabharwal, D.~P. Seetharam, Varying
  bandwidth resource allocation problem with bag constraints, in: IEEE
  International Symposium on Parallel \& Distributed Processing ({IPDPS}),
  2010, pp. 1--10.

\bibitem{spieksma1999approximability}
F.~C. Spieksma, On the approximability of an interval scheduling problem,
  Journal of Scheduling 2~(5) (1999) 215--227.

\bibitem{elbassioni2012approximation}
K.~Elbassioni, N.~Garg, D.~Gupta, A.~Kumar, V.~Narula, A.~Pal, {Approximation
  Algorithms for the Unsplittable Flow Problem on Paths and Trees}, in: IARCS
  Annual Conference on Foundations of Software Technology and Theoretical
  Computer Science (FSTTCS 2012), Vol.~18 of Leibniz International Proceedings
  in Informatics (LIPIcs), Dagstuhl, Germany, 2012, pp. 267--275.

\bibitem{Grandoni2015}
F.~Grandoni, S.~Ingala, S.~Uniyal, Improved Approximation Algorithms for
  Unsplittable Flow on a Path with Time Windows, Springer International
  Publishing, 2015, pp. 13--24.

\bibitem{nemirovski2008interior}
A.~S. Nemirovski, M.~J. Todd, Interior-point methods for optimization, Acta
  Numerica 17~(1) (2008) 191--234.

\bibitem{EN17}
K.~Elbassioni, T.~T. Nguyen, Approximation algorithms for binary packing
  problems with quadratic constraints of low cp-rank decompositions, Discrete
  Applied Mathematics 230 (2017) 56--70.

\bibitem{SR10}
B.~Patt-Shamir, D.~Rawitz, Vector bin packing with multiple-choice, in:
  Algorithm Theory - SWAT 2010, Springer Berlin Heidelberg, 2010, pp. 248--259.
\newblock \href {http://dx.doi.org/10.1007/978-3-642-13731-0_24}
  {\path{doi:10.1007/978-3-642-13731-0_24}}.

\bibitem{GLS88}
M.~Gr\"{o}tschel, L.~Lov\'{a}sz, A.~Schrijver, Geometric Algorithms and
  Combinatorial Optimization, Springer, New York, 1988.

\bibitem{S86}
A.~Schrijver, Theory of Linear and Integer Programming, Wiley, New York, 1986.

\bibitem{chakaravarthy2014improved}
V.~T. Chakaravarthy, A.~R. Choudhury, S.~Gupta, S.~Roy, Y.~Sabharwal, Improved
  algorithms for resource allocation under varying capacity, in: Algorithms-ESA
  2014, Springer Berlin Heidelberg, Berlin, Heidelberg, 2014, pp. 222--234.

\bibitem{chakrabarti2007approximation}
A.~Chakrabarti, C.~Chekuri, A.~Gupta, A.~Kumar, Approximation algorithms for
  the unsplittable flow problem, Algorithmica 47~(1) (2007) 53--78.

\bibitem{bar2001unified}
A.~Bar-Noy, R.~Bar-Yehuda, A.~Freund, J.~(Seffi)~Naor, B.~Schieber, A unified
  approach to approximating resource allocation and scheduling, J. ACM 48~(5)
  (2001) 1069--1090.
\newblock \href {http://dx.doi.org/10.1145/502102.502107}
  {\path{doi:10.1145/502102.502107}}.

\bibitem{IH78}
T.~Ibaraki, T.~Hasegawa, {The Multiple-Choice Knapsack Problem}, Journal of the
  Operations Research Society of Japan 21~(1) (1978) 59--93.

\bibitem{CHW75}
A.~Chandra, D.~Hirschberg, C.~Wong, Approximate algorithms for the knapsack
  problem and its generalizations, IBM Research Report RC56l6, IBM T. J. Watson
  Research Center.

\bibitem{KPP10book}
H.~Kellerer, U.~Pferschy, D.~Pisinger, Knapsack Problems, Springer, 2010.

\bibitem{AGLW13}
A.~Anagnostopoulos, F.~Grandoni, S.~Leonardi, A.~Wiese, Constant integrality
  gap {LP} formulations of unsplittable flow on a path, in: International
  Conference Integer Programming and Combinatorial Optimization ({IPCO}), 2013,
  pp. 25--36.

\bibitem{KKCE16b}
A.~Karapetyan, M.~Khonji, C.-K. Chau, K.~Elbassioni, Online algorithm for
  demand response with inelastic demands and apparent power constraint, Tech.
  rep., Masdar Institute, https://arxiv.org/abs/1611.00559 (2016).

\bibitem{MCK17}
M.~Khonji, S.~C.-K. Chau, K.~Elbassion, Combinatorial optimization of ac
  optimal power flow in radial distribution networks, arXiv preprint
  arXiv:1709.08431.

\bibitem{MCK18a}
M.~Khonji, S.~C.-K. Chau, K.~Elbassioni,
  \href{http://doi.acm.org/10.1145/3208903.3208934}{Challenges in scheduling
  electric vehicle charging with discrete charging rates in ac power networks},
  in: Proceedings of the Ninth International Conference on Future Energy
  Systems, e-Energy '18, 2018, pp. 183--186.
\newblock \href {http://dx.doi.org/10.1145/3208903.3208934}
  {\path{doi:10.1145/3208903.3208934}}.
\newline\urlprefix\url{http://doi.acm.org/10.1145/3208903.3208934}

\bibitem{MCK18b}
M.~Khonji, S.~C.-K. Chau, K.~Elbassioni,
  \href{http://doi.acm.org/10.1145/3208903.3213895}{Approximation scheduling
  algorithms for electric vehicle charging with discrete charging options}, in:
  Proceedings of the Ninth International Conference on Future Energy Systems,
  e-Energy '18, 2018, pp. 579--585.
\newblock \href {http://dx.doi.org/10.1145/3208903.3213895}
  {\path{doi:10.1145/3208903.3213895}}.
\newline\urlprefix\url{http://doi.acm.org/10.1145/3208903.3213895}

\bibitem{MCK18c}
M.~Khonji, S.~C.-K. Chau, K.~Elbassioni, Combinatorial optimization of electric
  vehicle charging in ac power distribution networks, in: IEEE International
  Conference on Communications, Control, and Computing Technologies for Smart
  Grids, SmartGridComm '18, 2018.

\end{thebibliography}

\section*{Appendix}

\subsection*{Proof of Lemma~\ref{lem:tb}}

\begin{customlem}{\ref{lem:tb}}[\cite{KKCMZ16}] 
	Given a set of 2D vectors $\{d_i \in \RR^2\}_{i=1}^n$
	$$ \frac{\sum_{i=1}^n |d_i| }{\bigg| \sum_{i =1}^n d_i \bigg|} \le {\sec \tfrac{\theta}{2}},$$
	where $\theta$ is the maximum angle between any pair of vectors and $0 \le \theta \le \frac{\pi}{2}$.
\end{customlem}
\begin{proof}
	If $\theta=0$ then the statement is trivial, therefore we assume otherwise. We prove $\frac{(\sum_{i=1}^n |d_i| )^2}{|\sum_{i=1}^n d_i |^2} \le \frac{2}{\cos \theta + 1}$  by induction (notice that $\sec \tfrac{\theta}{2} = \sqrt{\frac{2}{\cos \theta + 1}}$).  First, we expand  the left-hand side by
	\begin{align}
	&\frac{  \sum_{i=1}^n |d_i|^2 +  2\sum_{1\le i < j \le n} |d_i| \cdot |d_j|  } { \sum_{i =1}^n |d_i|^2 +  2\sum_{1\le i < j \le n} |d_i| \cdot |d_j| (\sin \theta_i \sin \theta_j + \cos \theta_i \cos \theta_j)}\notag\\
	&=\frac{ \sum_{i=1}^n |d_i|^2 +  2\sum_{1\le i < j \le n} |d_i| \cdot |d_j|  } { \sum_{i=1}^n |d_i|^2 +  2\sum_{1\le i < j \le n} |d_i| \cdot |d_j| \cos (\theta_i - \theta_j)}, \label{eq:ind}
	\end{align}
	where $\theta_i$ is the angle that $d_i$ makes with the $x$ axis.

	Consider the base case: $n= 2$. Eqn.~\raf{eq:ind} becomes
	\begin{align}
	\frac{|d_1|^2 +|d_2|^2 + 2 |d_1|\cdot |d_2|}{|d_1|^2 +|d_2|^2 + 2 |d_1|\cdot |d_2| \cos(\theta)} = f\Big(\frac{|d_2|}{|d_1|}\Big),
	\end{align}
	where $ f(x) \triangleq \frac{1+x^2+2x}{1+x^2+2x\cos \theta}$. The first derivative is given by
	$$f'(x) = \frac{(1+x^2+2x \cos \theta )(2x + 2) - 1+x^2 + 2x)(2x + 2 \cos \theta)}{(1 + x^2 + 2x \cos \theta )^2}$$ 
	$f'(x)$ is zero only when $x=1$. Hence, $f(1)$ is an extreminum point.  We compare $f(1)$ with $f(x)$ at the  boundaries $x\in \{0,\infty\}$:$$f(1) = \frac{2}{\cos \theta + 1} \ge f(0) = \lim_{x \to \infty} f(x) = 1$$
	Therefore, $f(x)$ has a global maximum of $\frac{2}{\cos \theta  + 1}$.
	
	Next, we proceed to the inductive step. We assume $\frac{\sum_{i=1}^{r-1} |d_i| }{\big| \sum_{i=1}^{r-1} d_i \big|} \le \sqrt{\frac{2}{\cos \theta + 1}}$ where $r \in \{1,\ldots, n\}$. W.l.o.g., assume $\theta_2\ge\theta_3\ge \cdots \ge \theta_n \ge \theta_1$. Rewrite Eqn.~\raf{eq:ind} as
	\begin{equation}
	\frac{ ( \sum_{i=1}^r |d_i|)^2 } { \displaystyle \sum_{i=1}^r |d_i|^2 + 2 {\sum_{1\le i<j<r}} |d_i|  |d_j| \cos (\theta_i - \theta_j) +  2 |d_r| {\sum_{1\le i<r}} |d_i|  \cos (\theta_i - \theta_r)} \label{eq:den}
	\end{equation}
	Let $g(\theta_r)$ be the denominator of Eqn.~\raf{eq:den}. We take the  second derivative of $g(\theta_r)$:
	$$
	g''(\theta_r) = -2 |d_r| \sum_{1 \le i < r} |d_i| \cos(\theta_i - \theta_r)
	$$ 
	Notice that $\cos(\theta_i - \theta_r) \ge 0$, therefore the second derivative is always negative. This indicates that all local exterma in $[0,\theta_{r-1}]$ of $g(\theta_n)$ are local maxima. Hence, the minimum occurs at the boundaries:
	$$
	\min_{\theta_r \in [0,\theta_{r-1}]}  g(\theta_r) \in \{g(0), g(\theta_{r-1})\}
	$$
	If $\theta_r \in \{0,\theta_r\}$ , then there must exist at least a pair of vectors in $\{d_i\}_{i=1}^r$ with the same angle. Combining these two vectors into one, we can obtain an instance with $r-1$ vectors. Hence, by the inductive hypothesis, the same bound holds up to $r$ vectors.
\end{proof}

\end{document}